\documentclass[
aps,prx,onecolumn,10pt,longbibliography,
  nofootinbib,superscriptaddress
]{revtex4-2}

\usepackage[T1]{fontenc}
\usepackage{newtxtext}
\usepackage{amsmath,mathtools}
\usepackage{amsthm}
\usepackage{newtxmath}
\usepackage{microtype}
\usepackage{enumitem}
\usepackage{graphicx}
\usepackage{xcolor}
\usepackage[colorlinks=true,linkcolor=teal,citecolor=teal,urlcolor=teal]{hyperref}

\hypersetup{pdftitle={Nonadditivity in fermionic quantum Shannon theory},
 pdfauthor={Farzin Salek, Amirreza Negari, Graeme Smith, Jens Eisert, Debbie Leung}}

\numberwithin{equation}{section}

\newtheorem{theorem}{Theorem}[section]
\newtheorem{proposition}[theorem]{Proposition}
\newtheorem{lemma}[theorem]{Lemma}
\newtheorem{corollary}[theorem]{Corollary}
\theoremstyle{definition}
\newtheorem{definition}[theorem]{Definition}

\theoremstyle{remark}
\newtheorem{remark}[theorem]{Remark}

\newcommand{\R}{\mathbb R}
\newcommand{\C}{\mathbb C}
\newcommand{\id}{\mathbf 1}
\newcommand{\ii}{\mathrm i}
\newcommand{\Tr}{\operatorname{Tr}}
\newcommand{\rank}{\operatorname{rank}}
\newcommand{\Pf}{\operatorname{Pf}}
\newcommand{\conv}{\operatorname{conv}}
\newcommand{\supp}{\operatorname{supp}}
\newcommand{\Span}{\operatorname{span}}
\newcommand{\diag}{\operatorname{diag}}
\newcommand{\fG}{\mathrm{fG}}

\newcommand{\Jman}{\mathcal J}
\newcommand{\eps}{\varepsilon}
\newcommand{\ket}[1]{|#1\rangle}
\newcommand{\bra}[1]{\langle #1|}
\newcommand{\braket}[2]{\langle #1|#2\rangle}
\newcommand{\ketbra}[1]{|#1\rangle\!\langle #1|}
\newcommand{\abs}[1]{\left|#1\right|}
\newcommand{\norm}[1]{\left\|#1\right\|}
\newcommand{\set}[1]{\left\{#1\right\}}

\newcommand{\transpose}{\mathsf T}
\newcommand{\ev}{\mathrm{even}}

\newcommand{\EF}{E_{\mathrm F}}
\newcommand{\EFalpha}{E_{\mathrm F,\alpha}}

\newcommand{\gA}{\gamma^{A}}
\newcommand{\gB}{\gamma^{B}}

\begin{document}

\title{Nonadditivity in fermionic quantum Shannon theory}
\affiliation{Perimeter Institute for Theoretical Physics,
Waterloo, Ontario N2L 2Y5, Canada}
\affiliation{Institute for Quantum Computing, University of Waterloo,
Waterloo, Ontario N2L 3G1, Canada}
\affiliation{Department of Applied Mathematics, University of Waterloo,
Waterloo, Ontario N2L 3G1, Canada}
\affiliation{Department of Combinatorics and Optimization, University of Waterloo,
Waterloo, Ontario N2L 3G1, Canada}

\affiliation{Dahlem Center for Complex Quantum Systems, Freie Universit\"at Berlin,
14195 Berlin, Germany}

\author{Farzin Salek}
\affiliation{Perimeter Institute for Theoretical Physics,
Waterloo, Ontario N2L 2Y5, Canada}
\affiliation{Institute for Quantum Computing, University of Waterloo,
Waterloo, Ontario N2L 3G1, Canada}
\affiliation{Dahlem Center for Complex Quantum Systems, Freie Universit\"at Berlin,
14195 Berlin, Germany}
\author{Amir-Reza Negari}
\affiliation{Perimeter Institute for Theoretical Physics,
Waterloo, Ontario N2L 2Y5, Canada}
\affiliation{Institute for Quantum Computing, University of Waterloo,
Waterloo, Ontario N2L 3G1, Canada}
\author{Graeme Smith}
\affiliation{Institute for Quantum Computing, University of Waterloo,
Waterloo, Ontario N2L 3G1, Canada}
\affiliation{Department of Applied Mathematics, University of Waterloo,
Waterloo, Ontario N2L 3G1, Canada}

\author{Jens Eisert}
\affiliation{Dahlem Center for Complex Quantum Systems, Freie Universit\"at Berlin,
14195 Berlin, Germany}

\author{Debbie Leung}
\affiliation{Perimeter Institute for Theoretical Physics,
Waterloo, Ontario N2L 2Y5, Canada}
\affiliation{Institute for Quantum Computing, University of Waterloo,
Waterloo, Ontario N2L 3G1, Canada}
\affiliation{Department of Combinatorics and Optimization, University of Waterloo,
Waterloo, Ontario N2L 3G1, Canada}

\begin{abstract}
Collective state preparation and channel encoding can reduce
entanglement requirements and increase communication rates. We study
these effects in fermionic systems under parity superselection, with
composition governed by the canonical anticommutation relations.
For a two-mode Gaussian state $\tau$, we prove nonadditivity of
entanglement of formation, obtaining
$\EF^{\mathrm{phys}}(\tau^{\otimes m})=\lceil m/2\rceil\log2$
for every $m\geq1$. Two copies therefore have the same formation
value as one, violating both additivity and strong superadditivity,
while regularization halves the single-copy value. For an associated
Gaussian channel, we determine the minimum output entropy and Holevo
information at every blocklength and prove nonadditivity of both.
Together, these results give exact answers to the fermionic
counterparts of the four additivity questions associated with
Shor and Hastings.
One physical use carries no classical information, whereas two uses
transmit one bit perfectly. The unassisted classical and zero-error
classical capacities both equal one half bit per mode. A sharp
finite-block bound gives a strong converse: the decoding success
probability vanishes at every rate above capacity.
These optimizations cover all physical pure-state decompositions
and channel inputs, with Gaussian ensembles and classical codes
attaining the optima.
For this state and channel, the formation and minimum output entropy
formulas hold for every R\'enyi order $\alpha\in(0,\infty]$.
We also determine the unassisted quantum and zero-error quantum
capacities, both equal to one half logical qubit per mode, and prove
a strong converse: entanglement fidelity vanishes at every rate
above capacity.
For a noisy Gaussian family, we determine the one- and two-copy
formation values and prove strict nonadditivity at full rank.
For the original state, Bell violation under local parity-preserving
measurements first appears at three copies.
\end{abstract}

\maketitle

\section{Introduction}
\label{sec:introduction}

Entanglement is a central resource in quantum information science that
renders quantum protocols more powerful than their classical counterparts.
Since the early days of the field, it has been an important question how
to meaningfully quantify entanglement \cite{PlenioVirmani2007} and thereby determine how much of
this useful resource 
\cite{ChitambarGour2019} is present in a given quantum state. Among these entanglement measures, the entanglement of formation takes
a particularly prominent role. It is defined as 
the minimum average entanglement in a
decomposition into pure quantum states of a mixed quantum state
\cite{Bennett1996MixedState,PlenioVirmani2007,Horodecki2009}.
Jointly decomposing several copies can reduce this average per copy.
In ordinary finite-dimensional quantum theory, the limiting value is
the entanglement cost~\cite{HaydenHorodeckiTerhal2001}. Additivity would
make that limit equal to the single-copy value.

The corresponding communication problem asks whether joint encoding
across channel uses can increase the Holevo information per use.
In ordinary finite-dimensional quantum theory, Shor has proved that the
universal additivity conjectures for entanglement of formation,
minimum output entropy, and Holevo information are equivalent to
strong superadditivity of formation \cite{Shor2004}.
Hastings has subsequently disproved these conjectures using random quantum
channels~\cite{Hastings2009}.
Here we have solved the analogous nonadditivity questions for a fermionic
example under parity superselection and composition governed by the
\emph{canonical anticommutation relations}  (CAR).
Starting from the entanglement of formation of a two-mode Gaussian quantum state,
we have determined the formation value, minimum output entropy, and Holevo
information at every blocklength. The optimizations cover all physical
decompositions into pure quantum states and channel inputs.
A common spectral constraint gives the converse
bounds, which explicit Gaussian ensembles and codes attain.
Related additivity questions for entanglement of purification and its
R\'enyi variants have led to explicit nonadditivity
examples~\cite{KhanianNegari2026Purification,NegariKhanian2026RenyiPurification}
and sufficient conditions for additivity at R\'enyi order
two~\cite{FarajiKhanian2026Renyi2Purification}.

Superselection rules restrict local state preparation and change how
entanglement and other nonlocal resources can be used
\cite{VerstraeteCirac2003,SchuchVerstraeteCirac2004Resources}.
In fermionic systems, parity superselection forbids a coherent
superposition of vacuum and occupation in an isolated mode. It does
allow correlations between odd operators on different modes, since their
product is even~\cite{PhysRevA.76.022311,Dariano2014,Szalay2021,kalarde2026fermionsfundamentallynonlocalbosons}.
Joint encoding can use these correlations even when each individual input
mode is uninformative. The canonical anticommutation relations 
specify how the correlations transform when uses are combined.
Products of quantum states and channel uses refer to this fermionic composition.
We denote it by $\otimes$ and explain its occupation-basis representation in
Sec.~\ref{sec:preliminaries}.

The distinction is operational: an isolated-mode description does not
specify the action on globally even products of odd operators from
different modes.  Our channel erases all physical single-mode signals
but retains such correlations across uses.  The resulting nonadditivity
is governed by the algebra of the retained Majoranas, which supports
both classical codes and noiseless logical quantum registers.

Our constructions use Gaussian states and channels.  A Gaussian state is entirely determined
by its Majorana two-point functions, and its dynamics under a
quadratic Hamiltonian can be described at the covariance-matrix level
\cite{Araki1970,Bravyi2005FLO,TerhalDiVincenzo2002}.
This description also gives modewise formulas for pure-state entanglement
\cite{BoteroReznik2004} and a characterization of Gaussian channels
\cite{GreplovaGiedke2018}.
Gaussian states and their superpositions are also used in approximation
algorithms for interacting fermions
\cite{BravyiGosset2017,BravyiGossetKoenigTemme2019}, and Gaussian
descriptions can emerge locally under suitable quadratic dynamics
\cite{GluzaKrumnowFriesdorfGogolinEisert2016}.
More recently, a covariance-matrix criterion for extendibility and finite
de Finetti bounds were obtained for fermionic Gaussian
states~\cite{NegariSalek2025}.
Covariance constraints also underlie recent product-state, de Finetti,
and recoverability bounds for fermionic Gaussian
systems~\cite{NegariEtAl2026Approximation}.

Mixed-state entanglement poses further difficulties.  The usual partial
transpose of a Gaussian density matrix, for example, need not be Gaussian,
which complicates calculations and bounds for negativity
\cite{EislerZimboras2015,EisertEislerZimboras2018}.
For formation, averaging covariance matrices is insufficient: an
ensemble must reproduce the full state, including its higher
correlations.  We construct such an ensemble explicitly and prove a
lower bound that does not assume Gaussian decomposition states.

To be specific, let
$\EF^{\mathrm{phys}}$ denote the formation roof over pure states of
definite total parity, and let $\EF^{\fG}$ restrict the decomposition to
Gaussian states.  For a state $\tau$ with one mode at 
a party referred to as
Alice and one at 
a party called
Bob,
we prove
\begin{equation}\label{eq:main-result-introduction}
 \EF^{\mathrm{phys}}(\tau^{\otimes m})
 =\EF^{\fG}(\tau^{\otimes m})
 =\left\lceil\frac m2\right\rceil\log2
 \qquad(m\geq1).
\end{equation}
Here $m$ is the number of copies jointly decomposed.  The single-copy
value is $\log2$, and the regularized value is $(\log2)/2$.
The construction fixes the correlated Majorana pairs and matches the
remaining Majoranas locally as far as possible.  Averaging over the signs
on the free pairs reproduces the full density operator.  A covariance-rank
bound first proves optimality among Gaussian decompositions.  For an
arbitrary physical pure state in the same support, the marginal commutes
with $m+1$ anticommuting
Hermitian unitaries.  These symmetries force each marginal eigenvalue to
occur at least $2^{\lceil m/2\rceil}$ times, giving the entropy lower bound
in Eq.~\eqref{eq:main-result-introduction} without a Gaussian assumption.

A Gaussian isometry maps input vectors onto the support of $\tau$.
Tracing out one output defines a channel $\mathcal N$, for which
entanglement entropy becomes output entropy~\cite{MatsumotoShimonoWinter2004}.
The same spectral argument
gives
\begin{equation}\label{eq:channel-results-introduction}
 S_{\min}^{\mathrm{phys}}(\mathcal N^{\otimes m})
   =\left\lceil\frac m2\right\rceil\log2,
 \qquad
 \chi^{\mathrm{phys}}(\mathcal N^{\otimes m})
   =\left\lfloor\frac m2\right\rfloor\log2.
\end{equation}
Here $m$ is the number of channel uses, or coding blocklength.
An explicit Gaussian code has orthogonal outputs, attains the Holevo
quantity, and achieves one half bit per mode with zero error.  The
spectral bound gives a sharp finite-block limit on the decoding
success of every physical code and a strong converse above this rate.
Thus the unassisted classical and zero-error capacities both equal
one half bit per mode.  The proof works directly with the fermionic
channel and does not invoke a fermionic extension of the general
equivalence theorem~\cite{Shor2004}.
The retained even Clifford algebra also yields exact quantum
transmission: a block of $m$ modes carries
$\max\{0,\lfloor(m-1)/2\rfloor\}$ logical qubits without error.
Section~\ref{sec:quantum-capacity} proves the asymptotic quantum
capacity of one half qubit per mode and a strong converse.

The same stabilizers give an exact threshold for Bell nonlocality.
With local parity-preserving measurements, one or two copies of $\tau$
cannot violate any Bell inequality, whereas three copies attain the
maximal quantum CHSH value $2\sqrt2$.  This corrects the multiple-copy
claim in Ref.~\cite[Sec.~V.1]{Debarba2020}, whose argument uses ordinary
qubit tensor powers in place of the fermionic product.
Proposition~\ref{prop:bell-activation-threshold} proves both the three-copy
violation and the impossibility with fewer copies.  Formation
nonadditivity and Bell nonlocality therefore appear at different copy
numbers for the same state.

The formation result also survives noise and is hence robust.  For
\begin{equation}
\tau_t=(\id-\ii t\gamma_2\gamma_3)/4,
\end{equation}
$0\leq t\leq1$, we find
both the one-copy and two-copy formation values.  They are strictly
nonadditive for every $t>0$, including the full-rank range $0<t<1$.
At full rank, the pure-state support argument is no longer available.  The local
Gaussian transformation instead reduces the problem to an active
one-mode-per-party state and uncorrelated local modes, whose contribution
to the convex roof can be removed exactly.

Several ingredients are established in earlier work.  The state $\tau$
and its maximal one-copy physical formation value were studied in
Refs.~\cite{Dariano2014,Dariano2014Feynman}.
The family $\tau_t$ 
has been
considered in
Ref.~\cite[Observation~4]{SpeeSchwaigerGiedkeKraus2018}, which 
has also given
the
two-copy Gaussian Bell-pair conversion.
That conversion gives an alternative description of the matching ensemble
and is used for the noisy family and the operational rates.
Our physical converse uses the Clifford-algebra structure that also
underlies bipartite stabilizer normal forms~\cite{AudenaertPlenio2005}.
Earlier minimum-output and channel-covariance calculations give
the Gaussian-restricted antecedents of
Eq.~\eqref{eq:channel-results-introduction}
\cite[Secs.~II.B and IV.A]{Bravyi2005Capacity}.
The direct matching and covariance arguments give the Gaussian formation
values.  The physical converse establishes their optimality over all
allowed decompositions and inputs, and the noisy calculation gives an
exact two-copy formation formula at full rank.
The further structure of this work is as follows. 
Section~\ref{sec:preliminaries} fixes the conventions and definitions.
We prove the formation theorem in Sec.~\ref{sec:nonadditivity} and derive
the communication results in Sec.~\ref{sec:channel}.  Supporting
convex-roof results and the full-rank calculation are given in the
appendices.

\section{Fermionic setting and entanglement measures}
\label{sec:preliminaries}

The formal statements made in this work require some preparation.
To this end, we collect the conventions and definitions used throughout
this work. General background on fermionic mode systems, parity
superselection, and fermionic Gaussian quantum states can be found in
Refs.~\cite{Araki1970,Bravyi2005FLO,PhysRevA.76.022311,
FriisLeeBruschi2013,Szalay2021}.

\subsection{Modes, parity, and bipartitions}

The notions of fermionic modes, parity, and bipartitions will feature
prominently in our analysis. We begin by fixing the notation for the modes
and their associated Majorana operators.
Let \(A\cup B\) be a bipartite system of \(n=n_A+n_B\) fermionic modes.
The \(2n\) Hermitian Majorana operators satisfy
\begin{equation}
  \{\gamma_j,\gamma_k\}=2\delta_{j, k}\id,
  \qquad 1\leq j,k\leq 2n .
\end{equation}
Our (commonly used) phase convention relating Majorana operators to annihilation and creation
operators is
\begin{equation}
\label{eq:majorana-from-c}
  \gamma_{2j-1}=c_j+c_j^\dagger,
  \qquad
  \gamma_{2j}=-\ii(c_j-c_j^\dagger),
\end{equation}
or, equivalently,
\begin{equation}
\label{eq:c-from-majorana}
  c_j=\frac12\left(\gamma_{2j-1}+\ii\gamma_{2j}\right),
  \qquad
  c_j^\dagger=\frac12\left(\gamma_{2j-1}-\ii\gamma_{2j}\right).
\end{equation}
Whenever an occupation-basis representation is required, we use the standard
Jordan--Wigner ordering
\begin{equation}
\label{eq:standard-jw-majoranas}
  c_j=Z_1\cdots Z_{j-1}\frac{X_j+\ii Y_j}{2},
\end{equation}
so that
\begin{equation}
\label{eq:standard-jw-majoranas-xy}
  \gamma_{2j-1}=Z_1\cdots Z_{j-1}X_j,
  \qquad
  \gamma_{2j}=Z_1\cdots Z_{j-1}Y_j .
\end{equation}
All sign-sensitive formulas below use this convention or explicitly state the
relevant mode ordering.
The total fermion-parity operator is
\begin{equation}
\label{eq:total-parity}
  P:=(-\ii)^n\gamma_1\gamma_2\cdots\gamma_{2n}.
\end{equation}
We impose fermion-parity superselection throughout.  Thus a physical density
operator satisfies
\begin{equation}
  [\rho,P]=0,
\end{equation}
while a physical pure state satisfies
\begin{equation}
  P\ket{\psi}=\pm\ket{\psi}.
\end{equation}
A physical mixed state may have support in both parity sectors, but it has no
coherence between them.  In particular, all odd Majorana moments vanish.

Alice and Bob control disjoint sets of modes.  A reduced state is the
restriction to the corresponding mode algebra.  To compute it by an
ordinary matrix partial trace, we use a Jordan--Wigner ordering adapted to
the bipartition, with all retained modes first.  Regrouping modes requires
fermionic swaps, including their occupation-dependent signs
\cite{SpeeSchwaigerGiedkeKraus2018,Szalay2021}.
Thus restriction is intrinsic, while a matrix partial trace in an arbitrary
interleaved ordering need not represent it correctly.
The entanglement of a physical bipartite pure state is
\begin{equation}
  E(\psi):=S(\psi_A),
  \qquad
  \psi_A:=\Tr_B\ketbra{\psi},
\end{equation}
where \(S(\omega):=-\Tr(\omega\log\omega)\).  All logarithms are natural, so a
maximally entangled two-level pair has entanglement \(\log 2\).

\subsection{Composition and local mode algebras}
\label{sec:fermionic-composition}

We next clarify how local mode algebras and physical quantum states
are combined to form a composite fermionic system.
We write $\rho\otimes\omega$ for the fermionic product of physical states
on disjoint mode sets.  Their embedded density operators are even and
commute, so this is a well-defined physical state; in a copy-contiguous
occupation ordering its matrix is the usual Kronecker product.
Regrouping the factors by party is a fermionic change of coordinates.
For two one-mode-per-party copies it reads
\begin{equation}\label{eq:fermionic-regrouping-sign}
 \ket{n_{A_1},n_{B_1},n_{A_2},n_{B_2}}
 \longmapsto
 (-1)^{n_{B_1}n_{A_2}}
 \ket{n_{A_1},n_{A_2},n_{B_1},n_{B_2}}.
\end{equation}
This sign is part of the representation of the same fermionic state;
it does not denote an additional physical interaction.
Tensor powers below always use this CAR composition, followed by the
Alice-versus-Bob bipartition.  In displayed Pauli matrices, $\otimes$
denotes the ordinary matrix tensor product in the explicitly stated
Jordan--Wigner ordering.

Global evenness does not require local evenness in every term of a
state's operator expansion.  If $a$ and $b$ are Majoranas in disjoint
mode sets, each is odd, but $\ii ab$ is even and can have nonzero
expectation in a physical state.  The full physical operator algebra on
several channel uses therefore contains correlations that cannot be
specified from the locally even operators alone.  A channel's action on
odd operators must also be fixed before forming its fermionic powers.
Section~\ref{sec:channel} gives an explicit example where this differs
from tensoring the channel's one-qubit matrix representation.  An explicit
calculation of the state regrouping is given in
Appendix~\ref{app:copy-composition}.

\subsection{Fermionic Gaussian states}

Fermionic Gaussian quantum states are the central objects of this work,
and we introduce them through their covariance matrices.
The covariance matrix of a parity-preserving quantum state \(\rho\) is the real
antisymmetric matrix
\begin{equation}
\label{eq:cm-definition}
  M_{j, k}:=\frac{\ii}{2}\Tr\rho[\gamma_j,\gamma_k],
  \qquad
  M^\transpose=-M .
\end{equation}
For \(j\neq k\), this convention gives
\begin{equation}
  M_{j, k}=\ii\Tr\rho\,\gamma_j\gamma_k,
  \qquad
  M_{j, j}=0 .
\end{equation}
A real antisymmetric matrix is the covariance matrix of a physical fermionic
state if and only if
\begin{equation}
\label{eq:bona-fide}
  \id+\ii M\geq0.
\end{equation}
Equivalently, all singular values of \(M\) lie in \([0,1]\), or
\(MM^\transpose\leq\id\).

\begin{definition}[Fermionic Gaussian state]
\label{def:fermionic-gaussian-state}
A parity-preserving state is fermionic Gaussian, or quasi-free, if its even
Majorana moments satisfy Wick's rule.  Explicitly, for every ordered even
index set \(I=\{i_1<\cdots<i_{2r}\}\),
\begin{equation}
\label{eq:wick-pfaffian}
  \Tr\rho(M)\,\ii^r\gamma_{i_1}\cdots\gamma_{i_{2r}}
  =\Pf(M_I),
\end{equation}
where \(M_I\) is the principal submatrix indexed by \(I\) and
\(\Pf(M_\varnothing)=1\).
\end{definition}

A fermionic Gaussian state is uniquely determined by its covariance matrix;
we denote it by \(\rho(M)\), or occasionally by \(\rho_{\rm G}(M)\).
Parity-preserving Gaussian unitaries act linearly on the Majorana operators:
an orthogonal transformation \(O\in SO(2n)\) sends
\begin{equation}
  M\longmapsto OMO^\transpose .
\end{equation}
Conversely, every \(O\in SO(2n)\) is implemented, up to phase, by a
parity-preserving Gaussian unitary.

A fermionic Gaussian state is pure if and only if
\begin{equation}
\label{eq:pure-cm-condition}
  M^2=-\id.
\end{equation}
We denote the compact manifold of pure fermionic Gaussian covariance matrices
by
\begin{equation}
\label{eq:pure-manifold}
  \Jman_{2n}
  :=
  \set{K\in\R^{2n\times2n}:
        K^\transpose=-K,\ K^2=-\id}.
\end{equation}

Every pure Gaussian state has definite total parity.  More strongly,
maximal full-system covariance forces an arbitrary state to be pure and
Gaussian; Appendix~\ref{app:gaussian-algebra} gives the elementary proof.

If \(K\in\Jman_{2n}\) is the covariance matrix of a bipartite pure Gaussian
state, its reduced state on \(A\) is Gaussian with covariance matrix \(K_A\),
the principal submatrix on Alice's Majoranas.  Let the eigenvalues of
\(\ii K_A\) be
\begin{equation}
  \pm s_1,\ldots,\pm s_{n_A},
  \qquad
  0\leq s_j\leq1.
\end{equation}
The entropy of entanglement is then
\begin{equation}
\label{eq:pure-entropy}
  E(K)
  =
  \sum_{j=1}^{n_A}
  h\!\left(\frac{1+s_j}{2}\right),
  \qquad
  h(x):=-x\log x-(1-x)\log(1-x).
\end{equation}
This is the fermionic modewise entropy formula
~\cite{BoteroReznik2003,BoteroReznik2004}.  Products of even Gaussian states
on disjoint mode sets are Gaussian, their covariance matrices are direct
sums, and partial trace corresponds to taking the appropriate principal
submatrix.

\subsection{Physical and Gaussian entanglement of formation}

To quantify the entanglement of fermionic quantum states, we introduce
the physical and Gaussian notions of entanglement of formation.
We use a physical convex roof compatible with parity superselection
~\cite{SchuchVerstraeteCirac2004,PhysRevA.76.022311,Dariano2014}.

\begin{definition}[Physical fermionic entanglement of formation]
\label{def:physical-eof}
Let \(\rho\) be a bipartite parity-preserving fermionic state.  Its physical
entanglement of formation is
given by
\begin{equation}
\label{eq:physical-fermionic-eof}
  \EF^{\mathrm{phys}}(\rho)
  :=
  \inf_{\substack{
    \rho=\sum_\alpha p_\alpha\ketbra{\psi_\alpha}\\
    p_\alpha\geq0,\ \sum_\alpha p_\alpha=1\\
    P\ket{\psi_\alpha}=\pm\ket{\psi_\alpha}}}
  \sum_\alpha p_\alpha
  S\!\left(\Tr_B\ketbra{\psi_\alpha}\right).
\end{equation}
The branches are arbitrary physical pure states, their parity eigenvalues may vary between branches.
\end{definition}

\begin{definition}[Fermionic Gaussian entanglement of formation]
\label{def:fgeof}
Let \(\rho(M)\) be a bipartite fermionic Gaussian state.  Its fermionic
Gaussian entanglement of formation is
\begin{equation}
\label{eq:fgeof-def}
  \EF^{\fG}(M)
  :=
  \inf\left\{
  \sum_\alpha p_\alpha E(K_\alpha)
  \;\middle|\;
  \begin{array}{l}
    \rho(M)=\sum_\alpha p_\alpha\rho(K_\alpha),\\
    p_\alpha\geq0,\quad \sum_\alpha p_\alpha=1,\\
    K_\alpha\in\Jman_{2n}
  \end{array}
  \right\}.
\end{equation}
The infimum is over finite pure-Gaussian ensembles.  We use
\(\EF^{\fG}(\rho(M))\) and \(\EF^{\fG}(M)\) interchangeably.  The feasible set
is nonempty by Theorem~\ref{thm:canonical-decomposition}.
The same pure-Gaussian roof defines \(\EF^{\fG}(\rho)\) for any state in
the convex hull of pure Gaussian states, including non-Gaussian targets.
We use this extension in the complete two-mode formula in
Appendix~\ref{sec:two-mode-formula}.
\end{definition}

We do not use the relaxed convention in which the target state commutes with
parity but individual pure branches may coherently mix the two parity
sectors.  Such alternatives are useful in other operational settings, but
they define a different convex roof
~\cite{SchuchVerstraeteCirac2004,PhysRevA.76.022311,NegariSalek2025}.
By Lemma~\ref{lem:pure-gaussian-parity}, every Gaussian branch admitted in
Definition~\ref{def:fgeof} is also physical.  Therefore, for every fermionic
Gaussian state,
\begin{equation}
\label{eq:physical-vs-gaussian-inequality}
  \EF^{\mathrm{phys}}(\rho(M))
  \leq
  \EF^{\fG}(\rho(M)).
\end{equation}

Both roofs are subadditive.  Indeed, products of admissible single-copy
ensembles are admissible for the product state, and hence
\begin{equation}
  \EF^\diamond(\rho\otimes\omega)
  \leq
  \EF^\diamond(\rho)+\EF^\diamond(\omega),
  \qquad
  \diamond\in\{\mathrm{phys},\fG\}.
\end{equation}
Subadditivity also ensures that their regularizations are well defined:
\begin{equation}
  \EF^{\diamond,\infty}(\rho)
  :=
  \lim_{m\to\infty}
  \frac{1}{m}\EF^\diamond(\rho^{\otimes m})
  =
  \inf_{m\geq1}
  \frac{1}{m}\EF^\diamond(\rho^{\otimes m}).
\end{equation}

\begin{lemma}[Support of ensemble branches]
\label{lem:support}
Let \(\rho\geq0\) be supported on a subspace \(S\).  If
\begin{equation*}
  \rho=\sum_\alpha p_\alpha\ketbra{\psi_\alpha},
  \qquad p_\alpha>0,
\end{equation*}
then every \(\ket{\psi_\alpha}\) belongs to \(S\).
\end{lemma}

\begin{proof}
For any \(\ket{\phi}\in S^\perp\),
\begin{equation*}
  0
  =
  \bra{\phi}\rho\ket{\phi}
  =
  \sum_\alpha
  p_\alpha
  \abs{\braket{\phi}{\psi_\alpha}}^2.
\end{equation*}
Every summand is nonnegative, so
\(\braket{\phi}{\psi_\alpha}=0\) for every \(\alpha\).  Hence
\(\ket{\psi_\alpha}\in S\).
\end{proof}

Finally, a convex average of covariance matrices need not reproduce the
corresponding Gaussian density operator.  If
\begin{equation}
  \rho(M)=\sum_\alpha p_\alpha\rho(K_\alpha),
\end{equation}
then necessarily
\begin{equation}
  M=\sum_\alpha p_\alpha K_\alpha,
\end{equation}
but the converse fails because covariance matrices contain only two-point
data.  Equality of the density operators additionally requires all higher
even moments to agree.  By Wick's rule, these are precisely the principal
Pfaffian constraints
\begin{equation}
  \Pf(M_I)
  =
  \sum_\alpha p_\alpha\Pf((K_\alpha)_I)
\end{equation}
for every even index set \(I\).  Appendix~\ref{sec:pfaffian-roof} uses these constraints to
derive the exact finite Pfaffian convex roof.

\section{Exact entanglement of formation}
\label{sec:nonadditivity}

With these notions in place, we turn to the exact entanglement of
formation of a two-mode fermionic quantum state that will serve as the
starting point for our analysis of nonadditivity.
Consider one fermionic mode at Alice and one at Bob, with the
Alice-first occupation convention fixed in Sec.~\ref{sec:preliminaries}.
We introduce the state vectors
\begin{equation}\label{eq:bell-even-odd}
  \ket{\Phi_+}:=\frac{\ket{0,0}+\ket{1,1}}{\sqrt2},
  \qquad
  \ket{\Psi_+}:=\frac{\ket{1,0}+\ket{0,1}}{\sqrt2},
\end{equation}
and define the quantum state
\begin{equation}\label{eq:tau-definition}
  \tau:=\frac12\ketbra{\Phi_+}+\frac12\ketbra{\Psi_+}.
\end{equation}
This state and its physical one-copy formation value were identified in
Ref.~\cite{Dariano2014}.  To see the one-copy obstruction, let
$\tau=\sum_\ell p_\ell\ketbra{\psi_\ell}$ be a physical decomposition.
By Lemma~\ref{lem:support}, every branch lies in
$\operatorname{span}\{\ket{\Phi_+},\ket{\Psi_+}\}$.
The two spanning vectors have opposite parity, so a vector in their span
has definite parity only if it is proportional to one of them.  Every physical
decomposition therefore uses the two Bell projectors with total weights
$1/2$.  Both have reduced state $\id_2/2$, giving
\begin{equation}\label{eq:tau-one-copy-value}
  \EF^{\mathrm{phys}}(\tau)=\EF^{\fG}(\tau)=\log2.
\end{equation}
Both Bell vectors are Gaussian.  The mixture is Gaussian as well:
with the Majoranas ordered as $(\gamma_1,\gamma_2)$ at Alice and
$(\gamma_3,\gamma_4)$ at Bob,
\begin{equation}
  Q:=-\ii\gamma_2\gamma_3=X\otimes X,
  \qquad
  \tau=\frac14(\id+Q).
  \label{eq:tau-projector-form}
\end{equation}
Its only nonzero independent covariance entry is
\begin{equation}\label{eq:tau-cm}
  (M_\tau)_{2,3}=-1.
\end{equation}
All other nonidentity Majorana moments vanish, as do the
corresponding Pfaffians, so Wick's rule holds.  One canonical
Majorana pair is pure and the other is maximally mixed.  We now ask
how the formation value changes when the decomposition may correlate
several copies.

For $m$ copies, let $\gA_1,\ldots,\gA_{2m}$ and
$\gB_1,\ldots,\gB_{2m}$ denote the global Majoranas at Alice and Bob.
Define
\begin{equation}\label{eq:Qr-definition}
  Q_j:=-\ii\gA_{2j}\gB_{2j-1},\qquad 1\leq j\leq m.
\end{equation}
Each $Q_j$ is a Hermitian unitary.  Distinct $Q_j$ use disjoint
pairs of Majoranas and commute: exchanging the two pairs requires
four anticommutations.  Applying Eq.~\eqref{eq:tau-projector-form}
to each copy therefore gives
\begin{equation}\label{eq:tau-m-product-projector}
  \tau_m:=\tau^{\otimes m}=2^{-2m}\prod_{j=1}^m(\id+Q_j).
\end{equation}
The operator
$\Pi:=\prod_{j=1}^m\frac{\id+Q_j}{2}$ is the orthogonal 
projection
onto
their simultaneous $+1$ eigenspace, and
$\tau^{\otimes m}=2^{-m}\Pi$.  In particular, this eigenspace is
exactly the support of the state.

\begin{theorem}[Exact formation value for every number of copies]
\label{thm:many-copy-nonadditivity}
\label{thm:physical-eof-nonadditivity}
For every integer $m\geq1$, with Alice holding the $m$ Alice modes
and Bob holding the $m$ Bob modes,
\begin{equation}\label{eq:tau-many-copy-value}
  \EF^{\mathrm{phys}}(\tau^{\otimes m})
  =\EF^{\fG}(\tau^{\otimes m})
  =\left\lceil\frac m2\right\rceil\log2.
\end{equation}
An optimal ensemble consists of pure Gaussian states whose reduced
states have a flat nonzero spectrum of rank
$2^{\lceil m/2\rceil}$.  In particular, for
$\mathsf X\in\{\mathrm{phys},\fG\}$,
\begin{equation}\label{eq:tau-two-copy-nonadditive}
  \EF^{\mathsf X}(\tau^{\otimes2})
  =\log2<2\log2=2\EF^{\mathsf X}(\tau).
\end{equation}
\end{theorem}

We first give an explicit Gaussian ensemble and prove its optimality
among Gaussian decompositions.  We then show that arbitrary physical
branches obey the same lower bound.

\subsection{The direct Majorana matching construction}
\label{sec:gaussian-nonadditivity}

We now give an explicit construction of the pure Gaussian quantum
states needed for achievability, based on a matching of Majorana operators.

\begin{proof}[Achievability in Theorem~\ref{thm:many-copy-nonadditivity}]
We start from the operators $Q_r=-\ii\gA_{2r}\gB_{2r-1}$
as in Eq.~\eqref{eq:Qr-definition}.
Equation~\eqref{eq:tau-m-product-projector} fixes these $m$ disjoint
pairs of Majorana operators across the bipartition in every branch.
We complete this partial matching to obtain the covariance matrix
of a pure Gaussian quantum state.
The remaining Majorana operators, which we refer to as free, are
\begin{equation}
  \mathcal F_A:=\{\gA_1,\gA_3,\ldots,\gA_{2m-1}\},
  \qquad
  \mathcal F_B:=\{\gB_2,\gB_4,\ldots,\gB_{2m}\}.
\end{equation}
Choose a perfect matching $\mathcal M$ that pairs as many free Majoranas
locally as possible.  It contains $\lfloor m/2\rfloor$ pairs inside $\mathcal F_A$, $\lfloor m/2\rfloor$ pairs inside $\mathcal F_B$, and, if $m$ is odd, one remaining cross pair between the last free Alice Majorana and the last free Bob Majorana.

For each pair $e\in\mathcal M$, choose an arbitrary orientation $e=(u_e,v_e)$ and set
\begin{equation}
  Q_e:=\ii u_ev_e.
\end{equation}
For each choice of signs $\sigma=(\sigma_e)_{e\in\mathcal M}\in\{\pm1\}^{\mathcal M}$, define the pure Gaussian state
\begin{equation}\label{eq:upper-branch-density}
  \rho_\sigma
  :=2^{-2m}
  \prod_{r=1}^m(\id+Q_r)
  \prod_{e\in\mathcal M}(\id+\sigma_e Q_e).
\end{equation}
The $m$ fixed pairs and the $m$ free pairs partition all $4m$ Majoranas
into $2m$ disjoint pairs.  Their bilinears are commuting Hermitian
unitaries, so Eq.~\eqref{eq:upper-branch-density} is the projector onto
a joint eigenspace.  In the occupation basis of these $2m$ normal modes,
each factor fixes one occupation.  The joint eigenspace is therefore
one-dimensional, and $\rho_\sigma$ is a normalized pure Gaussian state.
It has definite total parity, since all $2m$ occupations are fixed.
Its covariance matrix is a signed matching matrix and squares to $-\id$.

Averaging uniformly over the signs gives
\begin{align}\label{eq:upper-ensemble-average}
  2^{-m}\sum_{\sigma\in\{\pm1\}^{\mathcal M}}\rho_\sigma
  &=2^{-2m}\prod_{r=1}^m(\id+Q_r)
    \prod_{e\in\mathcal M}
    \left(\frac12\sum_{\sigma_e=\pm1}(\id+\sigma_e Q_e)\right)\notag\\
  &=2^{-2m}\prod_{r=1}^m(\id+Q_r)
   =\tau^{\otimes m}.
\end{align}
Thus \eqref{eq:upper-branch-density} gives a valid pure-Gaussian decomposition of $\tau^{\otimes m}$.

To compute the branch entanglement, observe that Alice's covariance block
contains exactly $\lfloor m/2\rfloor$ local Majorana pairs: those of
$\mathcal M$ that lie inside $\mathcal F_A$.  These contribute singular
values equal to $1$.  All remaining Alice Majoranas are paired across the
Alice--Bob cut and contribute zero singular values to her covariance block.
The canonical singular values of Alice's block, each counted once rather
than twice, are therefore
\begin{equation}
  \underbrace{1,\ldots,1}_{\lfloor m/2\rfloor\text{ times}},
  \underbrace{0,\ldots,0}_{\lceil m/2\rceil\text{ times}}.
\end{equation}
The entropy formula \eqref{eq:pure-entropy} gives
\begin{equation}\label{eq:upper-branch-entanglement}
  E(\rho_\sigma)
  =\left\lceil\frac m2\right\rceil\log 2
\end{equation}
for every sign choice $\sigma$.  The same covariance spectrum shows that Alice's reduced state is a
product of $\lfloor m/2\rfloor$ pure modes and $\lceil m/2\rceil$
maximally mixed modes, up to a local Gaussian unitary.  Its nonzero
spectrum is therefore flat, of rank $2^{\lceil m/2\rceil}$.  We obtain
\begin{equation}\label{eq:many-copy-upper-bound}
  \EF^{\mathrm{phys}}(\tau^{\otimes m})
  \leq\EF^{\fG}(\tau^{\otimes m})
  \leq\left\lceil\frac m2\right\rceil\log2.
\end{equation}
\end{proof}

\subsection{The Gaussian lower bound}

Here we establish a lower bound on the Gaussian entanglement of
formation of $\tau_m$ by exploiting the constraints imposed by its support.
In any decomposition of $\tau_m$, Lemma~\ref{lem:support} places every
branch in the support described above. For a Gaussian branch, the fixed
cross pairs force enough zero singular values in Alice's covariance
matrix to give the desired lower bound.

\begin{lemma}[A fixed Majorana bilinear fixes two covariance rows]
\label{lem:fixed-bilinear-row}
Let $K$ be the covariance matrix of a state supported on the $+1$
eigenspace of $Q=\epsilon\,\ii\gamma_u\gamma_v$, where $u\ne v$
and $\epsilon\in\{\pm1\}$.  Then
\begin{equation}\label{eq:fixed-bilinear-entry}
 K_{u,v}=\epsilon,\qquad
 K_{u,w}=K_{vw}=0\quad(w\notin\{u,v\}).
\end{equation}
\end{lemma}

\begin{proof}
The support condition gives
$\epsilon\ii\Tr\rho\gamma_u\gamma_v=1$, hence $K_{u,v}=\epsilon$.
Covariance admissibility implies $KK^\transpose\leq\id$, so each
row has squared norm at most one.  The entry $K_{u,v}$ already
saturates this bound for row $u$, as does $K_{v,u}=-\epsilon$ for
row $v$.  All other entries in those rows must vanish.
\end{proof}

The lemma does not require Gaussianity.  More generally, every Majorana
monomial containing exactly one of $\gamma_u,\gamma_v$ anticommutes
with $Q$ and has zero expectation in a state supported on its $+1$
eigenspace.  The next step uses Gaussianity to turn the covariance
spectrum into an entanglement bound.

\begin{lemma}[Rank lower bound for the many-copy support]\label{lem:many-copy-rank-lower-bound}
Let $K\in\Jman_{4m}$ be the covariance matrix of a pure Gaussian
state in the support of $\tau_m$, equivalently in the simultaneous
$+1$ eigenspace of the bilinears in Eq.~\eqref{eq:Qr-definition}.
Then the pure-branch entanglement obeys
\begin{equation}\label{eq:rank-lemma-entanglement-bound}
  E(K)\ge \left\lceil\frac m2\right\rceil\log 2.
\end{equation}
\end{lemma}

\begin{proof}
By Lemma~\ref{lem:fixed-bilinear-row}, each Alice Majorana $\gA_{2r}$ is paired only with Bob's Majorana $\gB_{2r-1}$.  Thus the Alice principal block $K_A$ has zero rows and columns on the span of $\gA_2,\gA_4,\ldots,\gA_{2m}$.  It can be nonzero only on the $m$-dimensional real subspace spanned by $\gA_1,\gA_3,\ldots,\gA_{2m-1}$.  An antisymmetric matrix on an $m$-dimensional real space has rank at most $2\lfloor m/2\rfloor$.  Therefore, among the $m$ canonical singular values $s_1,\ldots,s_m$ of $K_A$ entering the pure-Gaussian entropy formula, at least $m-\lfloor m/2\rfloor=\lceil m/2\rceil$ are zero.  Each zero singular value contributes $h(1/2)=\log 2$ to $E(K)$, proving \eqref{eq:rank-lemma-entanglement-bound}.
\end{proof}

Applying Lemma~\ref{lem:many-copy-rank-lower-bound} to each Gaussian
branch and using Eq.~\eqref{eq:many-copy-upper-bound} proves
$\EF^{\fG}(\tau_m)=\lceil m/2\rceil\log2$.
The matching construction and the lower bound describe the same
restriction: the fixed cross pairs leave only $m$ Majoranas at each
party available for local pairing.

\subsection{Optimality over all physical decompositions}
\label{sec:physical-nonadditivity}

To extend the lower bound beyond Gaussian quantum states, we show
that the reduced quantum state of every physical branch commutes
with a family of anticommuting operators, thereby forcing eigenvalue
multiplicities. The argument is closely related to the Clifford-algebra
description of bipartite stabilizer normal
forms~(see Ref.~\cite{AudenaertPlenio2005}); the following lemma gives
the spectral statement we need.

\begin{lemma}[Degeneracy from anticommuting symmetries]
\label{lem:clifford-degeneracy-bound}
Let $G_1,\ldots,G_k$ be pairwise anticommuting Hermitian unitaries
on a $D$-dimensional Hilbert space.  If a density matrix $\omega$
commutes with every $G_j$, then each eigenvalue of $\omega$ has
multiplicity at least $d=2^{\lfloor k/2\rfloor}$.  In particular,
\begin{equation}\label{eq:clifford-spectral-cap}
  \omega\leq\frac{\id_D}{d},
  \qquad S(\omega)\geq\log d.
\end{equation}
\end{lemma}

\begin{proof}
Put $q=\lfloor k/2\rfloor$ and
$F_j:=\ii G_{2j-1}G_{2j}$ for $1\leq j\leq q$.
Anticommutation gives $F_j^\dagger=F_j$ and $F_j^2=\id$.
For $j\ne\ell$, moving the two factors of $F_j$ past the two
factors of $F_\ell$ produces four minus signs, so
$F_jF_\ell=F_\ell F_j$.  Each $F_j$ also commutes with $\omega$.

Fix an eigenvalue $\lambda$ of $\omega$.  Its eigenspace is
invariant under all the $F_j$, and these commuting Hermitian
operators have a common normalized eigenvector $\ket v$ there
\begin{equation*}
 \omega\ket v=\lambda\ket v,
 \qquad F_j\ket v=s_j\ket v,
 \qquad s_j\in\{+1,-1\}.
\end{equation*}
For a binary string $x=(x_1,\ldots,x_q)$, define
\begin{equation}\label{eq:clifford-sign-pattern-vectors}
 \ket{v_x}:=G_1^{x_1}G_3^{x_2}\cdots G_{2q-1}^{x_q}\ket v.
\end{equation}
All these vectors have norm one and belong to the same eigenspace
of $\omega$, since every $G_j$ is unitary and commutes with
$\omega$.  The operator $G_{2j-1}$ anticommutes with $F_j$ and
commutes with $F_\ell$ for $\ell\ne j$.  It follows that
\begin{equation*}
 F_j\ket{v_x}=(-1)^{x_j}s_j\ket{v_x}.
\end{equation*}
If $x\ne y$, some $F_j$ therefore has different eigenvalues on
$\ket{v_x}$ and $\ket{v_y}$.  Since $F_j$ is Hermitian, the two
vectors are orthogonal.  We have constructed $2^q$ orthonormal
vectors with eigenvalue $\lambda$, proving the claimed multiplicity.
For $q=0$, the same argument consists simply of the one vector
$\ket v$.

Normalization now gives $d\lambda\leq1$ for every eigenvalue
$\lambda$ of $\omega$, hence the operator bound.  Finally,
\begin{equation*}
 S(\omega)
 =\sum_{\lambda>0}\lambda\log\frac1\lambda
 \geq\sum_{\lambda>0}\lambda\log d
 =\log d,
\end{equation*}
where the sums include multiplicities.  The operator $\id_D/d$ has
trace $D/d$, so the bound does not assert domination by a normalized
state when $D>d$.
\end{proof}

\begin{proof}[Converse for Theorem~\ref{thm:many-copy-nonadditivity}]
Take any physical decomposition
$\tau_m=\sum_\ell p_\ell\ketbra{\psi_\ell}$ with $p_\ell>0$.
By Lemma~\ref{lem:support} and Eq.~\eqref{eq:tau-m-product-projector},
every branch belongs to the simultaneous $+1$ eigenspace of the $Q_j$.
Fix one branch $\ket\psi$.  Physicality gives definite total parity, so
\begin{equation}\label{eq:physical-branch-stabilized}
  Q_j\ket\psi=\ket\psi\quad(1\leq j\leq m),
  \qquad P\ket\psi=\pm\ket\psi.
\end{equation}

Use the Jordan--Wigner ordering with all Alice modes before all
Bob modes.  If $a_1,\ldots,a_{2m}$ and
$b_1,\ldots,b_{2m}$ are the respective local Majorana matrices,
and $P_A$ is Alice's parity, the global Majoranas are represented by
\begin{equation}\label{eq:global-majoranas-tensor-representation}
  \gA_j=a_j\otimes\id_B,
  \qquad \gB_j=P_A\otimes b_j.
\end{equation}
Consequently,
\begin{equation}\label{eq:Alice-Clifford-generators}
  Q_j=A_j\otimes b_{2j-1},
  \qquad A_j:=-\ii a_{2j}P_A.
\end{equation}
These symmetries constrain Alice's marginal
$\psi_A:=\Tr_B\ketbra\psi$.  The relation $P_Aa_j=-a_jP_A$
implies $A_j^\dagger=A_j$ and $A_j^2=\id$.
Using Eq.~\eqref{eq:physical-branch-stabilized} and invariance of
partial trace under conjugation on Bob, we obtain
\begin{align}\label{eq:rhoA-commutes-Ar}
 \psi_A
 &=\Tr_B\bigl[(A_j\otimes b_{2j-1})\ketbra\psi
                  (A_j\otimes b_{2j-1})\bigr]\notag\\
 &=A_j\Tr_B\bigl[(\id_A\otimes b_{2j-1})\ketbra\psi
                  (\id_A\otimes b_{2j-1})\bigr]A_j\notag\\
 &=A_j\psi_A A_j.
\end{align}
Although the parity of $\ket\psi$ can have either sign,
$P\ketbra\psi P=\ketbra\psi$ in both cases.  As
$P=P_A\otimes P_B$, the same partial-trace argument gives
\begin{equation}\label{eq:rhoA-commutes-PA}
  P_A\psi_A P_A=\psi_A.
\end{equation}
Equations~\eqref{eq:rhoA-commutes-Ar} and
\eqref{eq:rhoA-commutes-PA} say that $\psi_A$ commutes with all
$m+1$ Hermitian unitaries $P_A,A_1,\ldots,A_m$.
Their anticommutation relations follow directly from the local CAR:
\begin{equation*}
 P_AA_j=\ii a_{2j}=-A_jP_A,
 \qquad
 A_jA_\ell=a_{2j}a_{2\ell}=-A_\ell A_j
 \quad(j\ne\ell).
\end{equation*}
We can therefore apply Lemma~\ref{lem:clifford-degeneracy-bound}
with $k=m+1$ and $D=2^m$.  Since
$\lfloor(m+1)/2\rfloor=\lceil m/2\rceil$, it gives
\begin{equation}\label{eq:branch-physical-spectral-cap}
  \psi_A\leq\frac{\id_{2^m}}{2^{\lceil m/2\rceil}},
  \qquad
  S(\psi_A)\geq\left\lceil\frac m2\right\rceil\log2.
\end{equation}
This holds for every physical pure branch, whether Gaussian or
not.  Averaging proves
\begin{equation}\label{eq:physical-lower-bound}
  \EF^{\mathrm{phys}}(\tau^{\otimes m})
  \geq\left\lceil\frac m2\right\rceil\log2,
\end{equation}
which, together with Eq.~\eqref{eq:many-copy-upper-bound}, completes
the proof.
\end{proof}

The distinction between the two lower bounds first matters at four
copies.  The state $\tau_m$ has $m$ maximally mixed normal modes and
$m$ pure normal modes.  By Theorem~\ref{thm:low-rank-equality} and
Lemma~\ref{lem:pure-spinor}, every physical pure state in its support
is already Gaussian when $m\leq3$.  For $m\geq4$, the support contains
non-Gaussian physical states, and these can occur with positive weight
in physical decompositions.  The Clifford converse proves that they
cannot lower the formation value.

The parity condition contributes one of the $m+1$ anticommuting
symmetries in the converse.  It is also essential for the one-copy
value: as an ordinary two-qubit matrix, $\tau$ has the separable
decomposition
\begin{equation}\label{eq:tau-unphysical-product-ensemble}
 \tau=\frac12\ketbra{+}_A\otimes\ketbra{+}_B
      +\frac12\ketbra{-}_A\otimes\ketbra{-}_B,
 \qquad \ket\pm:=\frac{\ket0\pm\ket1}{\sqrt2}.
\end{equation}
These product branches do not have definite total fermion parity
and are excluded from the physical roof
\cite{VerstraeteCirac2003,WisemanVaccaro2003,Dariano2014}.

\begin{corollary}[Regularization and R\'enyi orders]
\label{cor:regularized-tau}
\label{cor:renyi-nonadditive}
For $\alpha\in(0,\infty]$, let
$\EFalpha^{\mathsf X}$ use the same ensembles as
$\EF^{\mathsf X}$ and the branch cost $S_\alpha(\psi_A)$, where
\begin{equation}\label{eq:renyi-entropy-def}
 S_\alpha(\omega)=\frac{\log\Tr\omega^\alpha}{1-\alpha}
 \quad(\alpha\ne1,\infty),\qquad
 S_\infty(\omega)=-\log\norm\omega_\infty,
\end{equation}
and $S_1=S$.  For both $\mathsf X=\mathrm{phys}$ and
$\mathsf X=\fG$,
\begin{align}\label{eq:renyi-many-copy}
  \EFalpha^{\mathsf X}(\tau^{\otimes m})
   &=\left\lceil\frac m2\right\rceil\log2,\\
 \lim_{m\to\infty}\frac1m
  \EFalpha^{\mathsf X}(\tau^{\otimes m})
   &=\frac12\log2.
   \label{eq:regularized-tau}
\end{align}
\end{corollary}

\begin{proof}
Equation~\eqref{eq:branch-physical-spectral-cap} gives
$S_\infty(\psi_A)\geq\lceil m/2\rceil\log2$ for every
physical branch.  The inequality $S_\alpha\geq S_\infty$ follows
by writing $\Tr\omega^\alpha=\sum_j\lambda_j\lambda_j^{\alpha-1}$
and bounding the second factor by the appropriate power of the
largest eigenvalue; the direction reverses at $\alpha=1$, as
does the sign of $1-\alpha$.  Thus the same lower bound holds for
every order.  Every branch of the Gaussian construction has a
flat local spectrum of rank $2^{\lceil m/2\rceil}$, so it
attains this bound simultaneously for all $\alpha$.  Dividing by
$m$ gives the limit.
\end{proof}

In particular, no R\'enyi order restores additivity for this
fermionic example.  This includes $\alpha=2$, for which the
bosonic Gaussian formation quantity is additive and strongly
superadditive through its relation to Gaussian R\'enyi-$2$
squashed entanglement
\cite{Lami2016monogamy,Lami2017logdet}.
The regularized value also has an operational interpretation when
the physical two-mode Bell state is taken as the unit of entanglement.
Under fermionic
\emph{local operations with classical communication} (LOCC),
the formation cost and distillable entanglement
of $\tau$ are both one half Bell pair per copy.  The local conversion
in Eq.~\eqref{eq:pair-local-normal-form} gives the achievable rates.
Appendix~\ref{app:operational-rates} defines the allowed protocols and
proves the converse bounds, including for protocols with vanishing error.
The same two-copy state also violates strong superadditivity,
which would require
\begin{equation}\label{eq:physical-strong-superadditivity}
 \EF^{\mathsf X}(\omega_{A_1A_2:B_1B_2})
 \geq\EF^{\mathsf X}(\omega_{A_1:B_1})
      +\EF^{\mathsf X}(\omega_{A_2:B_2}).
\end{equation}
For $\omega=\tau_{A_1B_1}\otimes\tau_{A_2B_2}$, the left side is
$\log2$ and the right side is $2\log2$, for both roofs.

\subsection{An exact full-rank family}
\label{sec:full-rank-nonadditivity}

The local Gaussian transformation in
Appendix~\ref{sec:gaussian-matching-proof} also allows us to evaluate the
two-copy roof for a noisy family.  Consider the family already used in
Ref.~\cite{SpeeSchwaigerGiedkeKraus2018}, written as
\begin{equation}\label{eq:noisy-tau-definition}
 \tau_t:=\frac14(\id-\ii t\gamma_2\gamma_3),
 \qquad 0\leq t\leq1.
\end{equation}
Its canonical covariance parameters are $t$ and $0$, so it is
Gaussian.  Its eigenvalues are $(1+t)/4$ and $(1-t)/4$, each
with multiplicity two; in particular, it is full rank for $t<1$ and $\tau_1=\tau$.

\begin{theorem}[Exact two-copy nonadditivity at full rank]
\label{thm:full-rank-nonadditivity}
\label{cor:full-rank-nonadditivity}
For $\mathsf X\in\{\mathrm{phys},\fG\}$ and $0\leq t\leq1$,
\begin{align}\label{eq:noisy-one-copy-value}
 \EF^{\mathsf X}(\tau_t)
   &=f(t):=h\!\left(\frac{1+\sqrt{1-t^2}}2\right),\\
 \EF^{\mathsf X}(\tau_t^{\otimes2})
   &=g(t):=\frac{1+t^2}{2}
      h\!\left(\frac{t^2}{1+t^2}\right).
      \label{eq:noisy-two-copy-value}
\end{align}
For every $t>0$,
\begin{equation}\label{eq:noisy-strict-gap}
  g(t)<2f(t).
\end{equation}
\end{theorem}

The same local transformations as in Eq.~\eqref{eq:pair-local-SO-maps}
reduce two copies to an active two-mode Gaussian state with local
maximally mixed spectators.  The parity-sector formula then gives $g(t)$;
Appendix~\ref{app:fullrank-proof} constructs an optimal ensemble and proves
the matching converse.  The proof also gives the uniform gap
$2f(t)-g(t)\geq\tfrac12t^2\log2$.

For $t<1$, two copies have four mixed normal modes.  Their physical
pure-state decompositions can therefore contain non-Gaussian branches
already at two copies.  The lower bound in
Appendix~\ref{app:fullrank-proof} covers these branches as well.

Theorem~\ref{thm:full-rank-nonadditivity} establishes strict
nonadditivity for every $0<t<1$, including states arbitrarily close
to the maximally mixed state.  Unlike the endpoint $t=1$, the
formation value for more than two noisy copies is not determined
by this argument.

\begin{corollary}[An open set of nonadditive physical states]
\label{cor:open-nonadditivity}
For every fixed $0<t<1$, there is a trace-norm neighborhood of $\tau_t$
within the physical two-mode states on which
\begin{equation}\label{eq:open-nonadditivity}
 \EF^{\mathrm{phys}}(\rho^{\otimes2})
       <2\EF^{\mathrm{phys}}(\rho).
\end{equation}
\end{corollary}

\begin{proof}
The physical convex roof is a finite convex function on the
finite-dimensional convex set of parity-preserving quantum states.
It is therefore continuous on the relative interior, consisting of
the full-rank states.  Both $\tau_t$ and $\tau_t^{\otimes2}$ are
full rank.  Fermionic composition is continuous, so the gap
$2\EF^{\mathrm{phys}}(\rho)-\EF^{\mathrm{phys}}(\rho^{\otimes2})$
is continuous near $\tau_t$.  Its value there is strictly positive by
Theorem~\ref{thm:full-rank-nonadditivity}, and remains positive in a
sufficiently small neighborhood.
\end{proof}
This conclusion includes the Gaussian states in that neighborhood;
it concerns the physical roof and does not assert neighborhood
continuity or nonadditivity of the Gaussian-restricted roof.

\section{Exact communication consequences}
\label{sec:channel}
\label{sec:shor-violations}

The collective advantage in entanglement preparation has a striking
communication counterpart: channel uses that are individually
uninformative can transmit information with zero error when encoded
jointly.
The support of $\tau$ is the image of a one-mode input space under a
Gaussian isometry. If one output of this isometry is discarded, the
entanglement of an embedded input state vector becomes the entropy
at the remaining output. This observation relates the formation
calculation to a communication problem. The entropy bound from the
preceding section will give a converse, while a pair of Gaussian input
quantum states will attain it.

Each code state must be physical on the full input block.  It may
therefore correlate different channel uses, subject to total parity.
When we restrict to Gaussian encoding, each code state is required to
be Gaussian; its average need not be Gaussian.
For a parity-covariant channel $\mathcal L$, define
\begin{align}
 S_{\min}^{\mathrm{phys}}(\mathcal L)
 &:=\inf_{\ket\psi\,\mathrm{physical}}
       S\!\left(\mathcal L(\ketbra\psi)\right),
 \label{eq:physical-min-output-definition}\\
 \chi^{\mathrm{phys}}(\mathcal L)
 &:=\sup_{\{p_x,\rho_x\}}
 \left\{S\!\left(\mathcal L\!\left(\sum_xp_x\rho_x\right)\right)
       -\sum_xp_xS(\mathcal L(\rho_x))\right\}.
 \label{eq:physical-holevo-definition}
\end{align}
The first optimization is over unit vectors of definite total parity;
the second is over ensembles in which every $\rho_x$ commutes with
total parity.  Write $S_{\min}^{\fG}$ and $\chi^{\fG}$ when the respective
inputs are additionally required to be Gaussian.

For at most three input modes, every physical pure input is Gaussian
\cite[Lemma~3 and Sec.~III.E]{Bravyi2005Capacity}.
Refining each signal into parity-definite pure states cannot decrease
Holevo information, by concavity of output entropy.  Thus the physical
and Gaussian optima already coincide in these small dimensions.
Non-Gaussian pure inputs first become available on four modes; the
converse below applies to them as well.

\subsection{The channel and its fermionic composition}

To study collective encoding, we introduce the channel associated
with the quantum state $\tau$ and specify how its uses are combined
in the fermionic setting.
Let $C$ be one input mode and define
\begin{equation}\label{eq:gaussian-channel-isometry}
 V\ket0_C=\ket{\Phi_+}_{A,B},\qquad
 V\ket1_C=\ket{\Psi_+}_{A,B},\qquad
 V\frac{\id_C}{2}V^\dagger=\tau.
\end{equation}
Here $\ket{\Phi_+}=(\ket{0,0}+\ket{1,1})/\sqrt2$ and
$\ket{\Psi_+}=(\ket{1,0}+\ket{0,1})/\sqrt2$.
The isometry has the explicit Gaussian dilation
\begin{equation}\label{eq:channel-explicit-dilation}
 V\ket\psi=U(\ket\psi_A\ket0_B),\qquad
 U=\exp[-(\pi/4)\gamma_2\gamma_4]
   =\exp[-\ii(\pi/4)X_A\otimes Y_B],
\end{equation}
with the Jordan--Wigner ordering $A,B$.
For communication, $A$ is the receiver's output and $B$ is the
inaccessible environment.  The reduced channel
\begin{equation}\label{eq:gaussian-channel-definition}
 \mathcal N(\xi)=\Tr_B(V\xi V^\dagger)
\end{equation}
is therefore Gaussian.  The dilation also identifies the input operators
retained by the channel.  Conjugating the two output Majoranas gives
\[
 U^\dagger\gamma_1U=\gamma_1,
 \qquad U^\dagger\gamma_2U=-\gamma_4.
\]
The first belongs to the input, while the second belongs to the vacuum
environment and has zero expectation.  The channel retains the first
Majorana and erases the second.  In covariance-matrix notation,
\begin{equation}\label{eq:gaussian-channel-covariance-action}
 M\longmapsto D_1MD_1^{\transpose},\qquad D_1=\diag(1,0).
\end{equation}

Write $\Phi_m:=\mathcal N^{\otimes m}$ for the fermionic composition
of these dilations, with the tensor-product convention of
Sec.~\ref{sec:preliminaries}.  We must specify this composition on all
operators, since a covariance matrix alone does not determine a
non-Gaussian input.  Use the ordered input and output Majoranas
$\gamma_{2r-1}=Z_1\cdots Z_{r-1}X_r$ and
$\gamma_{2r}=Z_1\cdots Z_{r-1}Y_r$.  Under the combined dilation, an
odd-indexed output Majorana returns to the corresponding input
Majorana.  An even-indexed output Majorana returns, with a minus sign,
to the even-indexed Majorana of its own environment mode.  Conjugation takes a product to the product of these images.

Consider a Majorana monomial with distinct ordered indices.  If an
even index occurs, its dilation contains one Majorana from each of a
nonempty set of distinct environment modes.  This environmental factor
has zero expectation in the product vacuum.  If no even index occurs, the
environment contributes the identity.  The adjoint channel therefore has
the action displayed below.  Since the Majorana monomials form an orthogonal
operator basis and the surviving coefficients are zero or one, the
Schr\"odinger channel has the same action:
\begin{equation}\label{eq:channel-monomial-action}
 \Phi_m(\gamma_{j_1}\cdots\gamma_{j_s})=
 \begin{cases}
  \gamma_{j_1}\cdots\gamma_{j_s},&\text{if every $j_\ell$ is odd},\\
  0,&\text{otherwise},
 \end{cases}
 \qquad \Phi_m(\id)=\id.
\end{equation}
Here $j_1<\cdots<j_s$; repeated indices are removed using the CAR.
In particular, $M\mapsto D_mMD_m^{\transpose}$ with
$D_m=\diag(1,0,\ldots,1,0)$.
This is a special case of Bravyi's fermionic product
channels~\cite{Bravyi2005Capacity}, with attenuation coefficients zero or one.
An explicit completely positive representation is useful below.  Put
\begin{equation}\label{eq:channel-K-generators}
 P=(-\ii)^m\gamma_1\cdots\gamma_{2m},\qquad
 K_r=\ii P\gamma_{2r},\qquad
 \mathcal E_r(T)=\frac12(T+K_rTK_r).
\end{equation}
Each $K_r$ is a Hermitian unitary.  Conjugation by $K_r$ changes the
sign of $\gamma_{2r}$ and fixes every other Majorana.  Thus
$\mathcal E_r$ erases precisely the monomials containing $\gamma_{2r}$.
The maps $\mathcal E_r$ are trace preserving and completely positive,
and their diagonal actions on this basis commute.  Consequently,
\begin{equation}\label{eq:channel-conditional-expectation}
 \Phi_m=\mathcal E_1\circ\cdots\circ\mathcal E_m.
\end{equation}
This gives a completely positive representation of the full map, in
addition to its Gaussian dilation.

The combined isometry $V_m$ satisfies
$V_m(\id_{2^m}/2^m)V_m^\dagger=\tau^{\otimes m}$.  In particular,
its range is the support considered in the formation proof, and it
preserves the total parity of an input vector.  The output
$\Phi_m(\ketbra\psi)$ is the Alice marginal of
$V_m\ketbra\psi V_m^\dagger$.  The spectral bound below can therefore
also be read as the bound on the pure branches of
$\tau^{\otimes m}$.  We give its short channel proof to make clear that
it applies to every physical mixed input as well.

The fermionic composition matters even though every code state is
physical.  As a map on one-qubit matrices,
$\mathcal N(\xi)=(\xi+X\xi X)/2$.
However, its ordinary qubit tensor square annihilates $Y\otimes X$,
whereas $\Phi_2$ retains
$\ii\gamma_1\gamma_3=Y\otimes X$.
The latter observable is globally even: it is a product of one odd
Majorana from each mode.  For example,
$\rho_\pm=(\id_4\pm Y\otimes X)/4$ are physical states, and
$\Phi_2(\rho_\pm)=\rho_\pm$, while the qubit tensor square sends both
to $\id_4/4$.
The parity strings in the Jordan--Wigner representation account for this
difference~\cite{Bravyi2005Capacity,Szalay2021}.
Thus the single-mode matrix representation does not by itself specify
independent fermionic uses.  In what follows, $\Phi_m$ always means
the composition fixed by the dilation, or equivalently by
Eq.~\eqref{eq:channel-monomial-action}.

\subsection{The output bound and an optimal code}

To establish the optimality of a communication code, we first derive
a lower bound on the output entropy from a bound on the spectrum of
every physical channel output. We then construct a code that attains
the resulting communication bound.

\begin{lemma}[Output spectral bound]\label{lem:channel-output-bound}
For every physical input $\rho$ on $m$ modes,
\begin{equation}\label{eq:channel-output-operator-bound}
 \Phi_m(\rho)\leq\frac{\id_{2^m}}{d_m},\qquad
 d_m:=2^{\lceil m/2\rceil}.
\end{equation}
Consequently, $S(\Phi_m(\rho))\geq\log d_m$.
\end{lemma}
\begin{proof}
Write $\omega=\Phi_m(\rho)$.  Averaging $T$ and $K_rTK_r$ produces an
operator fixed by conjugation with $K_r$.  Since the averages
$\mathcal E_r$ commute, applying the other averages preserves this
property.  Hence $K_r\omega K_r=\omega$ for every $r$.
Moreover, each $K_r$ anticommutes with $P$, so each $\mathcal E_r$
commutes with conjugation by $P$.  The input identity $P\rho P=\rho$
therefore implies $P\omega P=\omega$.

We have obtained $m+1$ symmetries of the output.  They are pairwise
anticommuting: $PK_r=-K_rP$, and for $r\ne s$,
$K_rK_s=\gamma_{2r}\gamma_{2s}=-K_sK_r$.
Lemma~\ref{lem:clifford-degeneracy-bound} now gives multiplicity at least
$2^{\lfloor(m+1)/2\rfloor}=d_m$ for each occurring eigenvalue.
Unit trace implies $d_m\lambda\leq1$ for every eigenvalue $\lambda$,
which proves Eq.~\eqref{eq:channel-output-operator-bound}.
It also gives
$S(\Phi_m(\rho))=\sum_{\lambda>0}\lambda\log(1/\lambda)
\geq\log d_m$.
\end{proof}

\begin{theorem}[Exact minimum output entropy and Holevo information]
\label{thm:fermionic-shor-violations}
\label{thm:channel-exact-information}
For every $m\geq1$,
\begin{align}
 S_{\min}^{\mathrm{phys}}(\Phi_m)
 &=S_{\min}^{\fG}(\Phi_m)
 =\left\lceil\frac m2\right\rceil\log2,
 \label{eq:channel-minimum-output-many-copy}\\
 \chi^{\mathrm{phys}}(\Phi_m)
 &=\chi^{\fG}(\Phi_m)
 =\left\lfloor\frac m2\right\rfloor\log2.
 \label{eq:channel-holevo-many-copy}
\end{align}
An equiprobable ensemble of even pure Gaussian inputs attains the Holevo
optimum with pairwise orthogonal output supports; each of its codewords
also attains the minimum output entropy.
\end{theorem}
\begin{proof}
The output bound gives
$S_{\min}^{\mathrm{phys}}(\Phi_m)\geq\log d_m$.
For any physical ensemble, the average output has entropy at most
$m\log2$, while each conditional output has entropy at least
$\log d_m$.  Therefore its Holevo information is at most
$m\log2-\log d_m=\lfloor m/2\rfloor\log2$.
These bounds also hold under a Gaussian input restriction.  It remains
to construct a Gaussian ensemble attaining both of them.

For two uses, take $s\in\{+1,-1\}$ and the even pure Gaussian states
\begin{equation}\label{eq:channel-gaussian-codewords}
 \ket{\zeta_s}=\frac{\ket{0,0}+s\ii\ket{1,1}}{\sqrt2}.
\end{equation}
For example, these states are prepared from the vacuum by the
quadratic unitaries $\exp[(s\pi/4)\gamma_2\gamma_3]$.
To compute the output, expand the
quantum state
as
\[
 \ketbra{\zeta_s}
 =\frac14\bigl[\id_4+Z\otimes Z
                   +s(X\otimes Y+Y\otimes X)\bigr].
\]
Only the identity and $Y\otimes X=\ii\gamma_1\gamma_3$ survive:
$X\otimes Y=-\ii\gamma_2\gamma_4$ and $Z\otimes Z$ both contain
even-indexed Majoranas.  Hence
\begin{equation}\label{eq:channel-code-outputs}
 \omega_s:=\Phi_2(\ketbra{\zeta_s})
 =\frac14(\id_4+sY\otimes X).
\end{equation}
The operator $Y\otimes X$ has two eigenvalues, $+1$ and $-1$, each
of multiplicity two.  Thus $\omega_s$ is the normalized projector
onto its $s$ eigenspace.  The two outputs have orthogonal supports,
spectrum $\{1/2,1/2,0,0\}$, and uniform average $\id_4/4$.
The projective measurement
$\{(\id_4\pm Y\otimes X)/2\}$ recovers $s$ without error.
It is physical because the measured Majorana bilinear is even.
It also admits a Gaussian implementation: for
$R=\exp[(\pi/4)\gamma_2\gamma_3]$, one has
$R(Z\otimes\id_2)R^\dagger=\ii\gamma_1\gamma_3$.
Applying $R^\dagger$ and measuring the occupation of the first mode
therefore performs the required measurement.

For $m=2q+r$, where $r\in\{0,1\}$, encode $q$ independent signs
using Eq.~\eqref{eq:channel-gaussian-codewords} on disjoint input pairs.
If $r=1$, append $\ket0$ in the remaining mode; its output is
$\id_2/2$.  The $2^q$ resulting codewords are even, pure, and Gaussian.
Since each pair state is even, composition across disjoint pairs gives
the product of the corresponding pair outputs.  Different sign
strings differ on at least one pair and hence have orthogonal output
supports.  Each output is flat on a subspace of dimension
$2^q2^r=d_m$, so its entropy is $(q+r)\log2=\log d_m$.
Averaging independently over every sign gives
$\id_{2^m}/2^m$.  The conditional entropy and the entropy of the average
therefore attain the two bounds above.
\end{proof}

In particular, the exact two-use violations are
\begin{align}
 S_{\min}^{\mathrm{phys}}(\mathcal N^{\otimes2})
 &=\log2<2S_{\min}^{\mathrm{phys}}(\mathcal N)=2\log2,
 \label{eq:channel-minimum-output-violation}\\
 \chi^{\mathrm{phys}}(\mathcal N^{\otimes2})
 &=\log2>2\chi^{\mathrm{phys}}(\mathcal N)=0.
 \label{eq:channel-holevo-violation}
\end{align}
The same equalities hold for the Gaussian optimizations.
An isolated physical one-mode input is diagonal and is always sent to
$\id_2/2$; the two-use code places information in the even correlation
between the retained Majoranas.
The Gaussian minimum-output formula is already a specialization of the
parity-invariant calculation in Eq.~(27) of
Ref.~\cite{Bravyi2005Capacity}, whose channel-covariance argument also gives the
corresponding Gaussian Holevo optimum.  Bravyi's general Gaussian
optimization also allows parity-violating first moments and therefore
differs from the physical optimization considered here.
The spectral argument here establishes optimality against arbitrary
physical, including non-Gaussian, inputs and yields the following
finite-block converse.

A parallel-block code consists of a physical state on all $m$ input
modes for each message, followed by $\Phi_m$ and a joint measurement
of its outputs.  The sender and receiver share no prior entanglement.
The classical capacity is the supremum of asymptotically achievable
rates $(\log M_m)/m$ with vanishing average error; the zero-error
capacity requires zero error at every blocklength.

\begin{corollary}[Zero-error communication and a strong converse]
\label{cor:channel-capacity}
For $m$ modes, the largest zero-error alphabet has size
\begin{equation}\label{eq:channel-zero-error-alphabet}
 M_0(m)=2^{\lfloor m/2\rfloor}.
\end{equation}
Every unassisted code for $M$ equiprobable messages, with physical input
states and any decoding POVM, satisfies the sharp bound
\begin{equation}\label{eq:channel-success-bound}
 P_{\mathrm{succ}}
 \leq\min\left\{1,\frac{2^{\lfloor m/2\rfloor}}{M}\right\}.
\end{equation}
The unassisted classical capacity and zero-error capacity per input mode
are
\begin{equation}\label{eq:channel-capacity}
 C(\mathcal N)=C_0(\mathcal N)=\frac12\log2.
\end{equation}
Gaussian codewords and measurements of the pairwise Majorana bilinears
attain both capacities.
\end{corollary}
\begin{proof}
Let $D_x\geq0$ be the decoding operators, with
$\sum_{x=1}^M D_x=\id_{2^m}$, and let $\rho_x$ be the input for message
$x$.  The operator bound gives
\begin{equation}
 P_{\mathrm{succ}}
 =\frac1M\sum_x\Tr[D_x\Phi_m(\rho_x)]
 \leq\frac1{Md_m}\sum_x\Tr D_x
 =\frac{2^{\lfloor m/2\rfloor}}M.
\end{equation}
The code of Theorem~\ref{thm:channel-exact-information} attains
$M_0(m)$, and any subset gives a zero-error code of smaller size.
For $M>M_0(m)$, use the optimal code for $M_0(m)$ messages and assign
zero decoding operators to the remaining messages; this attains the
second branch of Eq.~\eqref{eq:channel-success-bound}.
Independent pairs achieve rate $(\log2)/2$ with zero error.
Conversely, if a sequence of codes satisfies
$\liminf_{m\to\infty}(\log M_m)/m>(\log2)/2$, the same bound forces
$P_{\mathrm{succ}}\to0$.  This proves both capacities and the strong
converse directly.
\end{proof}

\begin{remark}[R\'enyi minimum output entropies and output norms]
\label{rem:renyi-min-output}
For every $\alpha\in(0,\infty]$, define $S_{\min,\alpha}^{\mathrm{phys}}$
and $S_{\min,\alpha}^{\fG}$ by replacing $S$ with $S_\alpha$ in the
minimum-output definitions.  The eigenvalue bound and the flat output
spectra above give
\begin{equation}\label{eq:channel-renyi-minimum}
 S_{\min,\alpha}^{\mathrm{phys}}(\Phi_m)
 =S_{\min,\alpha}^{\fG}(\Phi_m)=\log d_m.
\end{equation}
Indeed, every R\'enyi entropy is at least
$S_\infty(\Phi_m(\rho))\geq\log d_m$, and the codewords attain equality.
For $1<p<\infty$, the corresponding maximal output Schatten norm is
\begin{equation}\label{eq:channel-output-p-norm}
 \nu_p^{\mathrm{phys}}(\Phi_m)
 :=\sup_{\ket\psi\,\mathrm{physical}}
 \norm{\Phi_m(\ketbra\psi)}_p
 =d_m^{1/p-1}=\nu_p^{\fG}(\Phi_m).
\end{equation}
At $p=\infty$ the value is $d_m^{-1}$.
Thus multiplicativity fails already for two uses at every $p>1$,
including $p=\infty$.
In the bosonic setting, tensor powers of a single-mode Gaussian channel
have multiplicative maximal output $p$-norms under Gaussian inputs for
$1<p<\infty$~\cite{SerafiniEisertWolf2005}.
\end{remark}

\subsection{Noiseless logical registers and quantum capacity}
\label{sec:quantum-capacity}

The correlations retained by the channel also carry quantum information.
We count logical qubits encoded in a fixed-total-parity subspace; both
encoding and decoding may use arbitrary operations with parity-preserving
dilations and physical local auxiliary systems.  No shared entanglement is
supplied. For a $D$-dimensional logical register, let
$\omega_{R\widehat L}$ be the joint state of the reference and decoded
register when the input is the normalized maximally entangled state
$\ket{\Omega_D}$. We use the entanglement fidelity
\[
 F_{\mathrm e}
 :=\bra{\Omega_D}\omega_{R\widehat L}\ket{\Omega_D}.
\]
The logical register and its reference are represented in
fixed-parity subspaces. Define
$Q(\mathcal N)$ and $Q_0(\mathcal N)$ using parallel blocks, respectively
with fidelity tending to one and with exact transmission, and rate
$(\log D)/m$.
The proof of this statement is presented in Appendix \ref{Proof}.

\begin{theorem}[Exact quantum transmission and a strong converse]
\label{thm:quantum-capacity}
For $m\geq1$, put
\begin{equation}\label{eq:quantum-logical-size}
 q_m:=\max\!\left\{0,\left\lfloor\frac{m-1}{2}\right\rfloor\right\}.
\end{equation}
The largest exactly transmitted logical register for $\Phi_m$ has
dimension $2^{q_m}$.  Every physical encoder and decoder transmitting
a $D$-dimensional logical register obeys
\begin{equation}\label{eq:quantum-fidelity-converse}
 F_{\mathrm e}\leq\min\{1,2^{q_m}/D\}.
\end{equation}
Consequently,
\begin{equation}\label{eq:quantum-capacity}
 Q(\mathcal N)=Q_0(\mathcal N)=\frac12\log2,
\end{equation}
or one half logical qubit per input mode.  Rates above this value have
entanglement fidelity tending to zero.
\end{theorem}

\section{Discussion}
\label{sec:conclusion}
For the quantum state $\tau$, collective preparation halves the
formation value per copy in the asymptotic limit. For the associated
channel, collective encoding changes zero one-use Holevo information
into a rate of one half bit per mode, attained with zero error and
accompanied by a strong converse. Both optimizations are fixed by the
same spectral degeneracy, and both are attained by Gaussian
constructions. Allowing arbitrary physical decompositions and classical signal
states does not improve these optima.

The present work has once again highlighted that collective resource
manipulation takes a distinctive and in some ways peculiar form in the Gaussian setting, limiting collective effects for Gaussian operations. This has
been particularly well understood for bosonic Gaussian quantum states:
the impossibility of entanglement distillation by local Gaussian
operations and classical communication has been established in
Refs.~\cite{EisertScheelPlenio2002,Fiurasek2002,GiedkeCirac2002}.
More recently, a no-go result for universal purification of bosonic
Gaussian quantum states, including approximate purification, has been
established for passive Gaussian operations in
Ref.~\cite{LiuEtAlPurification2025}. These results have shown how the
Gaussian structure shapes both the possibilities and importantly also the
striking limitations
of collective processing.

The same channel has quantum and zero-error quantum capacities of
one half logical qubit per mode. Its physical output algebra consists
of full logical matrix algebras with maximally mixed multiplicity
subsystems, and, at even blocklength, a classical sector label. This
algebraic structure explains both the exact quantum codes and their
strong converse. It identifies the logical degrees of freedom that can
carry information, showing how individually uninformative channel uses
can support both classical and quantum communication when combined.

The two-copy conversion also has an operational consequence. With a
physical two-mode Bell pair as the unit resource, $\tau$ can be formed
and distilled at a rate of one half Bell pair per copy. Appendix
\ref{app:operational-rates} proves these rates by composing the exact
conversion with a Schmidt-rank bound. The formation cost in entropy
units is therefore $(\log2)/2$, equal to the regularized physical roof.
The equality of formation and distillation rates gives an explicit
instance of asymptotically reversible entanglement manipulation under
parity superselection, relative to this unit resource.
For the noisy family $\tau_t$, the one- and two-copy formation values
have been determined exactly, but the many-copy values and their
regularization remain open. Determining them would show how the optimal
collective preparation changes away from the quantum state $\tau$.
The noisy family therefore extends this theoretical laboratory to a
controlled setting in which the robustness of the collective advantage
and the optimality of Gaussian decompositions can be investigated.
More generally, the Pfaffian constraints in
Appendix~\ref{sec:pfaffian-roof} give a finite formulation of the Gaussian
roof. An important next step is to determine whether this formulation
can yield useful approximation methods for larger systems, allowing the
questions resolved exactly here to be investigated beyond the solvable
example.

\section*{Acknowledgments}
We thank Sergey Bravyi and Antonio Anna Mele for helpful discussions.
F.S. acknowledges European Commission support through a
Marie Sk\l{}odowska-Curie Global Fellowship.
A.N. acknowledges support from the Natural Sciences and Engineering
Research Council of Canada (NSERC), Discovery Grant RGPIN-2018-04380,
and an Ontario Early Researcher Award.
G.S. acknowledges support from the NSERC--NSF Alliance program
(Grant No.~ALLRP-586858-2023) and the NSERC Discovery program
(Grant No.~RGPIN-2025-02094).
D.L. acknowledges support from NSERC (Grant No.~RGPIN-2024-03823).
J.E. acknowledges support from the BMFTR (MuniQC-Atoms, QuSol, and
Hybrid++), the Deutsche Forschungsgemeinschaft (CRC~183 and SPP~2514),
the Quantum Flagship programs MILLENION and PASQUANS2, Munich Quantum
Valley, Berlin Quantum, 
QuantERA,
and the European Research Council (DebuQC).

\appendix

\section{Canonical pure-Gaussian decompositions}
\label{sec:canonical-decomposition}

In this and the following appendices, we collect further details, background and proofs of the statements of the main text. We start with characterizing Gaussian pure quantum states. Every fermionic Gaussian state admits a finite pure-Gaussian decomposition:
its spectral decomposition in a Gaussian normal-mode basis.  We record this
standard construction~\cite{Araki1970,Bravyi2005FLO} to establish feasibility
of the Gaussian convex roof and fix the notation for the supporting results.

\begin{lemma}[Canonical form]
\label{lem:canonical-form}
For every physical \(n\)-mode fermionic Gaussian covariance matrix
\(M\in\R^{2n\times2n}\), there exist \(O\in SO(2n)\) and
\(\nu_1,\ldots,\nu_n\in[-1,1]\) such that
\begin{equation}
\label{eq:canonical-form}
  M
  =
  O\left(
    \bigoplus_{j=1}^n\nu_jJ
  \right)O^\transpose,
  \qquad
  J:=
  \begin{pmatrix}
    0&1\\
    -1&0
  \end{pmatrix}.
\end{equation}
\end{lemma}

\begin{proof}
Every real antisymmetric matrix is orthogonally block diagonalizable into
\(2\times2\) blocks proportional to \(J\).  The physicality condition
\eqref{eq:bona-fide} implies \(|\nu_j|\leq1\).  If the orthogonal matrix
obtained from the standard normal form has determinant \(-1\), composing it
with a reflection in one normal-mode plane changes the sign of the
corresponding \(\nu_j\) and produces an element of \(SO(2n)\).
\end{proof}

\begin{theorem}[Canonical Gaussian spectral decomposition]
\label{thm:canonical-decomposition}
Let \(\rho(M)\) be an \(n\)-mode fermionic Gaussian state with canonical form
\eqref{eq:canonical-form}.  For
\(\sigma=(\sigma_1,\ldots,\sigma_n)\in\{\pm1\}^n\), define
\begin{equation}
\label{eq:canonical-Ksigma}
  K_\sigma
  :=
  O\left(
    \bigoplus_{j=1}^n\sigma_jJ
  \right)O^\transpose,
  \qquad
  p_\sigma
  :=
  \prod_{j=1}^n
  \frac{1+\sigma_j\nu_j}{2}.
\end{equation}
Then \(K_\sigma\in\Jman_{2n}\), the coefficients \(p_\sigma\) form a
probability distribution, and
\begin{equation}
\label{eq:canonical-state-decomp}
  \rho(M)
  =
  \sum_{\sigma\in\{\pm1\}^n}
  p_\sigma\rho(K_\sigma).
\end{equation}

More precisely, if
\begin{equation}
  k
  :=
  \#\set{j:|\nu_j|<1}
\end{equation}
is the number of mixed normal modes, then exactly \(2^k\) probabilities are
nonzero.  The corresponding pure Gaussian states are mutually orthogonal,
and therefore
\begin{equation}
\label{eq:canonical-rank}
  \rank\rho(M)=2^k.
\end{equation}
\end{theorem}

\begin{proof}
First work in the canonical normal-mode basis, so that \(O=\id\), and define
the commuting occupation-sign observables
\begin{equation}
  Q_j
  :=
  \ii\widetilde\gamma_{2j-1}
  \widetilde\gamma_{2j}.
\end{equation}
The Gaussian state factorizes over its normal modes:
\begin{equation}
  \rho\left(
    \bigoplus_{j=1}^n\nu_jJ
  \right)
  =
  \prod_{j=1}^n
  \frac12\left(\id+\nu_jQ_j\right).
\end{equation}
Expanding each one-mode factor in its two occupation eigenstates gives
\begin{align}
  \rho\left(
    \bigoplus_{j=1}^n\nu_jJ
  \right)
  &=
  \sum_{\sigma\in\{\pm1\}^n}
  \left[
    \prod_{j=1}^n
    \frac{1+\sigma_j\nu_j}{2}
  \right]
  \prod_{j=1}^n
  \frac12\left(\id+\sigma_jQ_j\right)
  \nonumber\\
  &=
  \sum_{\sigma\in\{\pm1\}^n}
  p_\sigma
  \rho\left(
    \bigoplus_{j=1}^n\sigma_jJ
  \right).
\end{align}
Each covariance matrix \(\bigoplus_j\sigma_jJ\) squares to \(-\id\), so every
branch is pure Gaussian.  Conjugating the decomposition by the Gaussian
unitary implementing \(O\) proves Eq.~\eqref{eq:canonical-state-decomp}.

If \(|\nu_j|=1\), precisely one choice of \(\sigma_j\) has nonzero
probability; if \(|\nu_j|<1\), both choices occur.  Hence there are exactly
\(2^k\) nonzero branches.  They are joint eigenstates of the commuting
operators \(Q_j\), and are therefore mutually orthogonal.  This proves
Eq.~\eqref{eq:canonical-rank}.
\end{proof}

The finite canonical decomposition in Theorem~\ref{thm:canonical-decomposition}
establishes feasibility in Definition~\ref{def:fgeof}, without asserting
optimality.  In particular, unlike the
bosonic displacement ensemble, its branches need not have equal
entanglement.

\section{Gaussian operator expansions and covariance purity}
\label{app:gaussian-algebra}

This appendix collects the algebraic identities needed to expand
fermionic Gaussian quantum states in Majorana monomials and to
characterize their purity through covariance matrices.
For an even index set \(I=\{i_1<\cdots<i_{2r}\}\), define the Hermitian
Majorana monomial
\begin{equation}
\label{eq:monomial-basis}
  \Gamma_I:=\ii^r\gamma_{i_1}\cdots\gamma_{i_{2r}},
  \qquad
  \Gamma_\varnothing:=\id .
\end{equation}
The even Majorana monomials form
an orthogonal real basis of the Hermitian parity-even operators according to
\begin{equation}
\label{eq:majorana-orthogonality}
  \Tr(\Gamma_I \Gamma_J)=2^n\delta_{IJ}.
\end{equation}
Consequently, every parity-preserving state has the expansion
\begin{equation}
\label{eq:even-state-expansion}
  \rho
  =
  2^{-n}\sum_{I\ {\rm even}}\Tr(\rho \Gamma_I)\Gamma_I,
\end{equation}
and a Gaussian state has the Pfaffian expansion
\begin{equation}
\label{eq:gaussian-pfaffian-expansion}
  \rho(M)
  =
  2^{-n}\sum_{I\ {\rm even}}\Pf(M_I)\Gamma_I.
\end{equation}
This expansion will allow us to formulate the Gaussian convex roof as an
exact finite moment problem.
The following observation also fixes the relation between covariance purity
and parity.

\begin{lemma}[Maximal covariance and parity]
\label{lem:pure-gaussian-parity}
Let \(\rho\) be an arbitrary fermionic state of the full \(n\)-mode system
with covariance matrix \(M\).  If \(M^2=-\id\), then \(\rho\) is a pure
fermionic Gaussian state and has definite total parity.  In particular, every
pure fermionic Gaussian state is a parity eigenstate.
\end{lemma}

\begin{proof}
Choose \(O\in SO(2n)\) and signs \(\sigma_j\in\{\pm1\}\) such that
\begin{equation}
  M
  =
  O\left(\bigoplus_{j=1}^n\sigma_j J\right)O^\transpose,
  \qquad
  J=
  \begin{pmatrix}
    0&1\\
    -1&0
  \end{pmatrix}.
\end{equation}
In the corresponding Majorana basis, define
\begin{equation}
  Q_j:=\ii\widetilde\gamma_{2j-1}\widetilde\gamma_{2j}.
\end{equation}
The operators \(Q_j\) commute and satisfy \(Q_j^2=\id\), while the covariance
condition gives \(\Tr(\rho Q_j)=\sigma_j\).  Since \(Q_j\) has eigenvalues
\(\pm1\), the state is supported on the joint eigenspace
\(Q_j=\sigma_j\) for every \(j\).  This joint eigenspace is one-dimensional
because it fixes the occupation of every normal mode.  Hence \(\rho\) is the
projector onto a Gaussian occupation-number state, conjugated by the Gaussian
unitary implementing \(O\).  Moreover,
\(\prod_jQ_j=(-1)^nP\), so the state has definite total parity.
\end{proof}

The condition in Lemma~\ref{lem:pure-gaussian-parity} concerns the covariance
matrix of the full system.  A reduced covariance matrix need not satisfy
\(M^2=-\id\), even when the global state is pure.  Conversely, outside the
Gaussian family, purity alone does not force maximal covariance singular
values: a pure non-Gaussian state may carry information in higher-order
correlations that is absent from its covariance matrix.

For completeness, if a pure Gaussian branch has reduced covariance
parameters $s_1,\ldots,s_{n_A}$, the entropy formula extends to
\begin{equation}\label{eq:pure-gaussian-renyi-entropy}
 S_\alpha\!\left(\Tr_B\rho(K)\right)
 =\sum_{j=1}^{n_A}h_\alpha\!\left(\frac{1+s_j}{2}\right),
 \qquad
 h_\alpha(x):=\frac{\log(x^\alpha+(1-x)^\alpha)}{1-\alpha}.
\end{equation}
The endpoints $\alpha=1,\infty$ are understood by continuity.
This follows from the product spectrum of the reduced Gaussian state.
The two R\'enyi convex roofs use the same ensembles as their von Neumann
counterparts, so $\EFalpha^{\mathrm{phys}}\leq\EFalpha^{\fG}$ on Gaussian
targets.  More generally, Lemma~\ref{lem:clifford-degeneracy-bound} gives
$S_\alpha(\omega)\geq\lfloor k/2\rfloor\log2$ for every
$\alpha\in(0,\infty]$, since $S_\alpha\geq S_\infty$.

\section{The complete one-mode-versus-one-mode formula}
\label{sec:two-mode-formula}

For one mode held by Alice and one by Bob, the physical convex roof splits
into the two parity sectors.  This splitting is the formula established by
D'Ariano \emph{et al.}~\cite[Eq.~(13)]{Dariano2014}, within the fermionic
entanglement framework of Ref.~\cite{PhysRevA.76.022311}.  Each sector is
two-dimensional and every pure state in it is Gaussian.  We give the
resulting formula, its covariance dictionary, and explicit optimal
ensembles because they are used in the full-rank extension.
In order to proceed, we
use the block-ordered occupation basis
\begin{equation}\label{eq:block-basis}
  (\ket{0,0},\ket{1,1})\oplus(\ket{1,0},\ket{0,1}),
  \qquad
  \ket{1,1}=c_A^\dagger c_B^\dagger\ket{0,0}.
\end{equation}
Every parity-preserving two-mode state then has the form
\begin{equation}\label{eq:two-mode-block-state}
  \rho=
  \begin{pmatrix}
    a&z&0&0\\
    z^*&d&0&0\\
    0&0&b&w\\
    0&0&w^*&c
  \end{pmatrix},
\end{equation}
where
\begin{equation}
  a,b,c,d\ge0,
  \qquad
  a+b+c+d=1,
  \qquad
  ad\ge |z|^2,
  \qquad
  bc\ge |w|^2.
\end{equation}
Let
\begin{equation}\label{eq:sector-probs}
  p_e:=a+d,
  \qquad
  p_o:=b+c
\end{equation}
be the probabilities of the even- and odd-parity sectors.  When the
corresponding sector has nonzero weight, define
\begin{equation}\label{eq:sector-concurrences}
  C_e:=\frac{2|z|}{p_e},
  \qquad
  C_o:=\frac{2|w|}{p_o}.
\end{equation}
Positivity of \(\rho\) implies \(0\le C_e,C_o\le1\).  A sector of zero
weight is understood to contribute zero to all expressions below.

\begin{lemma}[Pure two-mode parity states are Gaussian]
\label{lem:two-mode-pure-gaussian}
Every pure state of definite parity on two fermionic modes is
Gaussian.
\end{lemma}

\begin{proof}
An even-parity vector has the form
\begin{equation}
  \alpha\ket{0,0}+\beta\ket{1,1}.
\end{equation}
When \(\alpha\ne0\), it is proportional to the BCS state
\begin{equation}
  \exp\!\left(
    \frac{\beta}{\alpha}c_A^\dagger c_B^\dagger
  \right)\ket{0,0};
\end{equation}
the limiting cases are occupation-number states and are Gaussian as
well.  An odd-parity vector can be written as
\begin{equation}
  \alpha\ket{1,0}+\beta\ket{0,1}
  =
  \left(\alpha c_A^\dagger+\beta c_B^\dagger\right)\ket{0,0}.
\end{equation}
This is a one-particle Slater determinant and hence is also Gaussian.
\end{proof}

\begin{theorem}[Two-mode parity-sector formula]
\label{thm:two-mode-formula}
For every parity-preserving state of the form
\eqref{eq:two-mode-block-state},
\begin{align}\label{eq:two-mode-fgeof}
  \EF^{\mathrm{phys}}(\rho)
  =\EF^{\fG}(\rho)
  &=
  p_e\,h\!\left(
    \frac{1+\sqrt{1-C_e^2}}{2}
  \right)
  \notag\\
  &\quad+
  p_o\,h\!\left(
    \frac{1+\sqrt{1-C_o^2}}{2}
  \right).
\end{align}
Here the Gaussian roof is understood on the convex hull of pure Gaussian
states, as in Definition~\ref{def:fgeof}.  Every physical two-mode state
belongs to that convex hull by Lemma~\ref{lem:two-mode-pure-gaussian}; the
target itself need not be Gaussian.  In particular, the formula gives both
roofs for every one-mode-versus-one-mode Gaussian state.
\end{theorem}

\begin{proof}
Every physical pure state has definite total parity.  Hence any
admissible pure-state decomposition of \(\rho\) separates into even-
and odd-parity branches.  Projecting the decomposition onto the two
parity sectors shows that their total weights must be \(p_e\) and
\(p_o\), respectively.  When these weights are nonzero, the
corresponding normalized states are
\begin{equation}
  \rho_e
  =\frac{1}{p_e}
  \begin{pmatrix}
    a&z\\
    z^*&d
  \end{pmatrix},
  \qquad
  \rho_o
  =\frac{1}{p_o}
  \begin{pmatrix}
    b&w\\
    w^*&c
  \end{pmatrix}.
\end{equation}
Conversely, pure-state decompositions of \(\rho_e\) and \(\rho_o\)
can be combined with weights \(p_e\) and \(p_o\) to produce a
decomposition of \(\rho\).  Therefore the physical convex roof splits
as
\begin{equation}\label{eq:sector-splitting-cost}
  \EF^{\mathrm{phys}}(\rho)
  =
  p_e \EF(\rho_e)+p_o \EF(\rho_o).
\end{equation}

Viewed as two-qubit states, \(\rho_e\) and \(\rho_o\) are supported on
\(\Span\{\ket{0,0},\ket{1,1}\}\) and
\(\Span\{\ket{1,0},\ket{0,1}\}\), respectively.  Their concurrences are
\begin{equation}
  C(\rho_e)=\frac{2|z|}{p_e}=C_e,
  \qquad
  C(\rho_o)=\frac{2|w|}{p_o}=C_o.
\end{equation}
Wootters' formula~\cite{Wootters1998} consequently gives
\begin{equation}
  \EF(\rho_s)
  =
  h\!\left(
    \frac{1+\sqrt{1-C_s^2}}{2}
  \right),
  \qquad s\in\{e,o\}.
\end{equation}
Substitution into \eqref{eq:sector-splitting-cost} proves
\eqref{eq:two-mode-fgeof} for the physical convex roof.

Finally, Lemma~\ref{lem:two-mode-pure-gaussian} shows that every pure
state appearing in either parity sector is Gaussian.  The physical
and Gaussian admissible ensembles therefore coincide, which proves
\(\EF^{\mathrm{phys}}(\rho)=\EF^{\fG}(\rho)\).
\end{proof}

An optimal ensemble can be written explicitly in each parity sector of
nonzero weight.  In the ordered basis of that sector,
write its normalized state as
\[
 R=\begin{pmatrix}u&v\\v^*&1-u\end{pmatrix},\qquad
 s=\frac{1+\sqrt{1-4|v|^2}}2.
\]
For $v\ne0$, put $\theta=\arg v$ and define
\begin{align}
 \ket{\phi_+}&=\sqrt{s}\ket{\eta_0}
       +e^{-\ii\theta}\sqrt{1-s}\ket{\eta_1},\\
 \ket{\phi_-}&=\sqrt{1-s}\ket{\eta_0}
       +e^{-\ii\theta}\sqrt{s}\ket{\eta_1}.
\end{align}
Here $\ket{\eta_0},\ket{\eta_1}$ are the ordered occupation basis
vectors of the chosen parity sector.
If $s>1/2$, then
\begin{equation}\label{eq:two-mode-explicit-ensemble}
 R=r\ketbra{\phi_+}+(1-r)\ketbra{\phi_-},\qquad
 r=\frac{u-(1-s)}{2s-1}.
\end{equation}
Positivity gives $u(1-u)\ge |v|^2$, hence $1-s\le u\le s$ and
$0\le r\le1$.  Both branches have entanglement $h(s)$, so this ensemble
attains the sector contribution in Eq.~\eqref{eq:two-mode-fgeof}.
For $v=0$, use the occupation-state decomposition; for $s=1/2$, positivity
forces $u=1/2$ and $R$ is already pure.  Combining the two sector ensembles
with their weights $p_e,p_o$ gives at most four pure Gaussian branches.

For a Gaussian input, all quantities in
Theorem~\ref{thm:two-mode-formula} can be obtained directly from the
covariance matrix.

\begin{proposition}[Covariance form for Gaussian states]
\label{prop:two-mode-cm-dictionary}
Order the Majorana operators as
\begin{equation}
  \gamma_1=c_A+c_A^\dagger,
  \qquad
  \gamma_2=-\ii(c_A-c_A^\dagger),
  \qquad
  \gamma_3=c_B+c_B^\dagger,
  \qquad
  \gamma_4=-\ii(c_B-c_B^\dagger).
\end{equation}
For a two-mode Gaussian state \(\rho(M)\), define
\begin{equation}\label{eq:pf-four}
  \Delta:=\Pf M
  =M_{1,2}M_{3,4}-M_{1,3}M_{2,4}+M_{1,4}M_{2,3}.
\end{equation}
Then
\begin{equation}\label{eq:pf-parity}
  p_e=\frac{1+\Delta}{2},
  \qquad
  p_o=\frac{1-\Delta}{2}.
\end{equation}
Whenever the corresponding denominator is nonzero, the sector
concurrences are
\begin{align}\label{eq:two-mode-cm-concurrences}
  C_e
  &=
  \frac{
    \sqrt{
      (M_{1,4}+M_{2,3})^2+(M_{1,3}-M_{2,4})^2
    }
  }{1+\Delta},
  \\
  C_o
  &=
  \frac{
    \sqrt{
      (M_{1,4}-M_{2,3})^2+(M_{1,3}+M_{2,4})^2
    }
  }{1-\Delta}.
\end{align}
Together with \eqref{eq:two-mode-fgeof}, these expressions give a
closed covariance-matrix formula for both
\(\EF^{\mathrm{phys}}(\rho(M))\) and \(\EF^{\fG}(\rho(M))\).
\end{proposition}

\begin{proof}
The total parity operator is
\begin{equation}
  P=\ii^2\gamma_1\gamma_2\gamma_3\gamma_4.
\end{equation}
It has eigenvalue \(+1\) on the even sector and \(-1\) on the odd
sector.  Wick's theorem therefore gives
\begin{equation}
  p_e-p_o
  =\Tr\!\bigl[\rho(M)P\bigr]
  =\Pf M=\Delta,
\end{equation}
which, together with \(p_e+p_o=1\), proves
\eqref{eq:pf-parity}.
Using
\begin{equation}
  c_A=\frac12(\gamma_1+\ii\gamma_2),
  \qquad
  c_B=\frac12(\gamma_3+\ii\gamma_4),
\end{equation}
and the covariance convention \eqref{eq:cm-definition}, a direct
expansion gives
\begin{align}
  z
  &=
  \frac14\left[
    -(M_{1,4}+M_{2,3})
    -\ii(M_{1,3}-M_{2,4})
  \right],
  \\
  w
  &=
  \frac14\left[
    (M_{1,4}-M_{2,3})
    +\ii(M_{1,3}+M_{2,4})
  \right].
\end{align}
Taking absolute values and substituting into
\eqref{eq:sector-concurrences}, together with
\eqref{eq:pf-parity}, yields
\eqref{eq:two-mode-cm-concurrences}.
\end{proof}

\section{Proof of the full-rank formation formulas}
\label{app:fullrank-proof}

In this appendix, we derive the full-rank formation formulas by
reducing the optimization to the even and odd parity sectors, where
the two-qubit entanglement-of-formation formula applies.

\begin{proof}
For one mode per party, every physical pure vector is Gaussian.
The physical roof separates into the even and odd parity sectors,
each of which has the two-qubit formation value given by
Wootters' formula
\cite{Wootters1998,PhysRevA.76.022311,Dariano2014}; we recall
the formula and its proof in Appendix~\ref{sec:two-mode-formula}.
For $\tau_t$ each sector has weight $1/2$ and concurrence $t$,
proving Eq.~\eqref{eq:noisy-one-copy-value} for both roofs.
For two copies, apply exactly the local Gaussian transformations
in Eq.~\eqref{eq:pair-local-SO-maps}.  The two active cross
covariances are now $-t$ and all spectator covariances vanish.
Gaussian uniqueness gives
\begin{equation}\label{eq:noisy-local-normal-form}
 \tau_t^{\otimes2}
 \simeq_{\mathrm{local\ Gaussian}}
 \sigma_t\otimes\frac{\id_{A_s}}2\otimes\frac{\id_{B_s}}2,
\end{equation}
where, in the Alice-first occupation representation of the two
active modes,
\begin{equation}\label{eq:noisy-active-state}
 \sigma_t=\frac14\bigl(\id+tX\otimes X-tY\otimes Y
                                    +t^2Z\otimes Z\bigr).
\end{equation}
Its even sector has weight $(1+t^2)/2$ and normalized
concurrence $2t/(1+t^2)$; its odd sector is diagonal.
The two-mode formula therefore gives
\begin{equation}\label{eq:noisy-active-roof}
 \EF^{\mathrm{phys}}(\sigma_t)
 =\EF^{\fG}(\sigma_t)=g(t).
\end{equation}
An optimal Gaussian ensemble of $\sigma_t$ is obtained by setting
\begin{equation}\label{eq:noisy-active-ensemble-vectors}
 \ket{u_t}:=\frac{\ket{0,0}+t\ket{1,1}}{\sqrt{1+t^2}},
 \qquad
 \ket{v_t}:=\frac{t\ket{0,0}+\ket{1,1}}{\sqrt{1+t^2}}.
\end{equation}
Then
\begin{align}\label{eq:noisy-active-optimal-ensemble}
 \sigma_t={}&\frac{1+t^2}{4}
      \bigl(\ketbra{u_t}+\ketbra{v_t}\bigr)\notag\\
 &+\frac{1-t^2}{4}
      \bigl(\ketbra{01}+\ketbra{10}\bigr).
\end{align}
The first two branches each have entanglement
$h(t^2/(1+t^2))$, while the other two are product states.
Tensoring this ensemble with the occupation ensembles of the
spectators and applying the inverse local Gaussian unitaries
proves the upper bound $g(t)$ for both roofs of
$\tau_t^{\otimes2}$.

To see that the spectator modes cannot lower the physical roof,
take any physical pure-state ensemble of the state on the right
of Eq.~\eqref{eq:noisy-local-normal-form} and measure both
spectator occupations locally.  Every nonzero conditional branch on the
active modes is pure and has definite parity: the total branch parity
and both measured occupation parities are fixed.  The active branches
are understood after fermionic restriction: deleting fixed occupation
spectators introduces only local parity unitaries in the occupation
representation and leaves each branch's entanglement unchanged.
The spectator occupation projectors commute with active observables,
so these branches form a physical ensemble of $\sigma_t$.  Average pure-state
entanglement cannot increase under the two local measurements,
by concavity of the unmeasured party's entropy at each step.
The original ensemble thus has average entanglement at least
$\EF^{\mathrm{phys}}(\sigma_t)=g(t)$.  This proves the matching
physical lower bound; the Gaussian lower bound follows from
$\EF^{\fG}\geq\EF^{\mathrm{phys}}$.

It remains to prove strictness for the entire family.  Let
$\Delta(t):=2f(t)-g(t)$ and $r:=\sqrt{1-t^2}$.  For $0<t<1$,
direct differentiation gives
\begin{align}\label{eq:noisy-gap-derivative}
 \frac{\Delta'(t)}t
 &=\frac{2\operatorname{arctanh}r}{r}
       -\log\frac{1+t^2}{t^2}\notag\\
 &=-\left(\frac1r-1\right)\log t^2
      +\frac2r\log(1+r)-\log(1+t^2).
\end{align}
The first term in the second line is nonnegative.  Concavity of
$\log(1+r)$ gives $\log(1+r)\geq r\log2$ on $[0,1]$,
and $\log(1+t^2)\leq\log2$.  Hence
$\Delta'(t)\geq t\log2$.  Since $\Delta(0)=0$,
$\Delta(t)\geq\tfrac12t^2\log2>0$ for $t>0$; the
endpoint $t=1$ follows by continuity.
\end{proof}

\section{The local Gaussian conversion}
\label{sec:gaussian-matching-proof}

The matching ensemble in Sec.~\ref{sec:gaussian-nonadditivity} has an
alternative description in local occupation coordinates.  The
conversion of two copies into an active Bell pair and local spectator
modes was established by Spee, Schwaiger, Giedke, and Kraus
\cite[Observation~4]{SpeeSchwaigerGiedkeKraus2018}.  We give its explicit
form in our Majorana convention.  The same change of basis is used for
the noisy family in Appendix~\ref{app:fullrank-proof} and the operational
rates in Appendix~\ref{app:operational-rates}.

\begin{proof}
For two copies, denote the global Majoranas at Alice by
$\gA_1,\ldots,\gA_4$ and those at Bob by
$\gB_1,\ldots,\gB_4$.  The only correlated pairs are
$(\gA_2,\gB_1)$ and $(\gA_4,\gB_3)$, each with covariance $-1$.
Make the local changes of Majorana basis
\begin{align}\label{eq:pair-local-SO-maps}
 (\alpha_1,\alpha_2,\alpha_3,\alpha_4)
   &=(\gA_2,\gA_4,\gA_1,-\gA_3),\notag\\
 (\beta_1,\beta_2,\beta_3,\beta_4)
   &=(\gB_3,\gB_1,\gB_2,\gB_4).
\end{align}
Both signed permutations have determinant $+1$, so each is
implemented by a local parity-preserving Gaussian unitary.  The
Majoranas $(\alpha_1,\alpha_2)$ define an active Alice mode and
$(\beta_1,\beta_2)$ an active Bob mode; the two remaining Majoranas
at each party define a spectator mode.  In the new basis the only
nonzero independent covariance entries are
$M_{\alpha_1,\beta_2}=M_{\alpha_2,\beta_1}=-1$.
In the order $(\alpha_1,\alpha_2,\beta_1,\beta_2)$, the active
covariance matrix is
\begin{equation}\label{eq:active-bell-covariance}
 M_{\mathrm{act}}=
 \begin{pmatrix}
 0&0&0&-1\\
 0&0&-1&0\\
 0&1&0&0\\
 1&0&0&0
 \end{pmatrix}.
\end{equation}
This is the covariance matrix of $\ket{\Phi_+}$ in our convention.
For example, its entries $M_{1,4}=M_{2,3}=-1$ follow from
$\langle Y\otimes Y\rangle_{\Phi_+}=-1$ and
$\langle X\otimes X\rangle_{\Phi_+}=1$.  The covariance entries
involving a spectator Majorana all vanish.  The full covariance
matrix is therefore the direct sum of $M_{\mathrm{act}}$ and two
zero $2\times2$ matrices, after grouping the active modes.
A zero one-mode covariance matrix represents $\id_2/2$.
Because the state is Gaussian and its covariance matrix determines it
uniquely, this proves the local normal form
\begin{equation}\label{eq:pair-local-normal-form}
 \tau^{\otimes2}
 \simeq_{\mathrm{local\ Gaussian}}
 \ketbra{\Phi_+}_{A_0B_0}
 \otimes\frac{\id_{A_s}}2\otimes\frac{\id_{B_s}}2.
\end{equation}
Here $\simeq_{\mathrm{local\ Gaussian}}$ denotes equivalence under
local parity-preserving Gaussian unitaries, with the fermionic
product convention of Sec.~\ref{sec:preliminaries}.

Let $U=U_AU_B$ be a product of local Gaussian unitaries implementing
Eq.~\eqref{eq:pair-local-normal-form}, so that $U\tau^{\otimes2}U^\dagger$
is the state on its right-hand side.  Decomposing the two spectator
states in their occupation bases gives four vectors.  In the
Alice-first order $A_0,A_s,B_0,B_s$, they are
\begin{equation}\label{eq:pair-optimal-spectator-ensemble}
 \ket{\xi_{j, k}}
 =\frac{\ket{0,j,0,k}+(-1)^j\ket{1,j,1,k}}{\sqrt2},
 \qquad j,k\in\{0,1\}.
\end{equation}
The sign $(-1)^j$ comes from moving the occupied $B_0$ mode past
$A_s$; it is the regrouping sign of
Eq.~\eqref{eq:fermionic-regrouping-sign}.  Equivalently,
$\ketbra{\xi_{j, k}}$ is the fermionic product of the active Bell
projector and the two occupation projectors.  Hence
\begin{equation}\label{eq:pair-inverse-unitary-ensemble}
 \tau^{\otimes2}
 =\frac14\sum_{j,k=0}^1\ketbra{\psi_{j, k}},
 \qquad \ket{\psi_{j, k}}:=U^\dagger\ket{\xi_{j, k}}.
\end{equation}

This ensemble is admissible for both convex roofs.  Each
$\ket{\xi_{j, k}}$ is a product of pure Gaussian states on disjoint
mode sets and has total parity $(-1)^{j+k}$; both properties are
preserved by $U^\dagger$.  Its Alice marginal is
\begin{equation}
 \Tr_{B_0B_s}\ketbra{\xi_{j, k}}
 =\frac12\bigl(\ketbra{0,j}+\ketbra{1,j}\bigr).
\end{equation}
It has two nonzero eigenvalues, both $1/2$.  The same is true for
$\ket{\psi_{j, k}}$, since a local unitary leaves the reduced spectrum
unchanged.  Thus all four branches have entanglement $\log2$.

For $m=2q+r$, where $r\in\{0,1\}$, apply this construction to each
of the $q$ pairs of copies.  If $r=1$, decompose the remaining copy
using Eq.~\eqref{eq:tau-definition}.  There are
$4^q2^r=2^m$ equally weighted branches.  Products of definite-parity
Gaussian pure states remain Gaussian and have definite total parity,
so these branches are again admissible.  Before undoing the local
unitaries, each branch has $q+r$ active Bell pairs and $q$ pure
spectator modes at each party.  Its Alice covariance matrix is a
direct sum of $q+r$ zero one-mode blocks and $q$ pure one-mode
blocks.  The marginal consequently has $2^{q+r}$ equal nonzero
eigenvalues, each $2^{-(q+r)}$, and entropy $(q+r)\log2$.
Undoing the local unitaries preserves this spectrum.  Consequently,
the ensemble attains the upper bound in
Eq.~\eqref{eq:many-copy-upper-bound}.
\end{proof}

To compare with the direct construction, choose the free matching pairs
$(\gA_1,\gA_3)$ and $(\gB_2,\gB_4)$ for a pair of copies.  In the
coordinates of Eq.~\eqref{eq:pair-local-SO-maps}, their bilinears satisfy
\begin{equation}
 \ii\alpha_3\alpha_4=-\ii\gA_1\gA_3,
 \qquad
 \ii\beta_3\beta_4=\ii\gB_2\gB_4.
\end{equation}
The bilinear of a mode with occupation $j$ has eigenvalue $2j-1$.
Consequently, the matching signs of
Eq.~\eqref{eq:upper-branch-density} are
$\sigma_A=(-1)^j$ and $\sigma_B=-(-1)^k$ for the spectator occupations
$j,k$.  Reversing the orientation of a matching pair reverses its sign
label.  Thus uniform averaging over the free local matching signs is
exactly uniform averaging over the spectator occupations.  Applying
this correspondence to disjoint pairs of copies identifies the two
ensembles at every even copy number.  At an odd copy number, the
remaining free cross-pair sign gives the two Bell branches of
Eq.~\eqref{eq:tau-definition}.

\section{Fermionic copy composition and Bell nonlocality}
\label{app:copy-composition}

The distinction between fermionic and ordinary qubit regrouping can be
seen directly in the state used throughout 
this work.
In the
copy-contiguous ordering $A_1,B_1,\ldots,A_m,B_m$, its density matrix is
the Kronecker power of $(\id_4+X\otimes X)/4$.  The unitary coordinate
change to $A_1,\ldots,A_m,B_1,\ldots,B_m$ acts on occupation strings as
\begin{equation}\label{eq:general-copy-regrouping}
 \ket{\boldsymbol n_A,\boldsymbol n_B}_{\mathrm{interleaved}}
 \longmapsto
 (-1)^{\sum_{r<s}n_{B_r}n_{A_s}}
 \ket{\boldsymbol n_A}\otimes\ket{\boldsymbol n_B}.
\end{equation}
Every sign records a crossing of occupied creation operators.
This is the fermionic swap convention of
Refs.~\cite{SpeeSchwaigerGiedkeKraus2018,Szalay2021}.

For two copies in the grouped ordering $A_1,A_2,B_1,B_2$, the stabilizers
of Eq.~\eqref{eq:Qr-definition} are
\begin{equation}\label{eq:two-copy-pauli-stabilizers}
 Q_1=X\otimes Z\otimes X\otimes\id,
 \qquad
 Q_2=\id\otimes X\otimes Z\otimes X.
\end{equation}
Consequently, the matrix of the fermionic product is
\begin{align}\label{eq:two-copy-grouped-matrix}
 \tau^{\otimes2}
 =\frac1{16}\bigl(&\id_{16}
       +X\otimes Z\otimes X\otimes\id\notag\\
       &+\id\otimes X\otimes Z\otimes X
       +X\otimes Y\otimes Y\otimes X\bigr).
\end{align}
If the crossing sign were omitted, ordinary qubit regrouping would instead
give
\begin{align}\label{eq:two-copy-unsigned-matrix}
 \frac1{16}\bigl(&\id_{16}
       +X\otimes\id\otimes X\otimes\id\notag\\
       &+\id\otimes X\otimes\id\otimes X
       +X\otimes X\otimes X\otimes X\bigr).
\end{align}
These matrices describe different correlations across the grouped cut.
Debarba, Iemini, Giedke, and Friis~\cite[Sec.~V.1]{Debarba2020}
argued that no number of copies of their two-mode mixed maximally
entangled state, which is $\tau$ here, can violate a Bell inequality
under parity superselection.  Their argument invokes separability of
the ordinary two-qubit representative and its ordinary tensor powers.
The inference does not apply to the fermionic product: the
regrouping signs distinguish
Eq.~\eqref{eq:two-copy-grouped-matrix} from
Eq.~\eqref{eq:two-copy-unsigned-matrix}.
The following proposition gives a counterexample and determines the
minimum number of copies.

\begin{proposition}[Exact threshold for Bell violation]
\label{prop:bell-activation-threshold}
Consider $\tau^{\otimes m}$ with the CAR composition, using only these
copies and local measurements whose effects commute with the total
parity at the respective party.
For $m=1,2$, every Bell experiment admits a local hidden-variable model.
For every $m\geq3$, local Gaussian unitaries followed by occupation
measurements attain the maximal quantum CHSH value $2\sqrt2$,
without postselection.
\end{proposition}

\begin{proof}
First consider the impossibility with fewer than three copies.
Define the local parity dephasing map
\begin{equation}\label{eq:local-parity-dephasing}
 \mathcal T_{\mathrm{loc}}(\rho)
 =\frac14\sum_{u,v=0}^1
 P_A^uP_B^v\,\rho\,P_B^vP_A^u.
\end{equation}
Every allowed local measurement effect commutes with the relevant
parity.  Its joint outcome probabilities are therefore identical on
$\rho$ and $\mathcal T_{\mathrm{loc}}(\rho)$.

For one copy,
$\mathcal T_{\mathrm{loc}}(\tau)=\id_4/4$.
For two copies, each $Q_r$ is odd at both parties, so local parity
dephasing removes the terms $Q_1$ and $Q_2$ while retaining their
product, since
\begin{equation}\label{eq:two-copy-parity-dephased}
 \mathcal T_{\mathrm{loc}}(\tau^{\otimes2})
 =\frac{\id_{16}+Q_1Q_2}{16}.
\end{equation}
In the Alice-first matrix representation,
Eq.~\eqref{eq:two-copy-pauli-stabilizers} gives
\begin{equation*}
 Q_1Q_2=(X\otimes Y)_A\otimes(Y\otimes X)_B.
\end{equation*}
Consequently,
\begin{equation}\label{eq:two-copy-local-model}
 \mathcal T_{\mathrm{loc}}(\tau^{\otimes2})
 =\frac12\sum_{s=\pm1}
 \left(\frac{\id_4+sX\otimes Y}{4}\right)_A
 \otimes
 \left(\frac{\id_4+sY\otimes X}{4}\right)_B.
\end{equation}
Each factor is a normalized positive state commuting with local
parity $Z\otimes Z$.  Thus the dephased state is separable.
The displayed decomposition gives a local hidden-variable model,
with shared variable $s$, for every allowed Bell experiment on two
copies.  The one-copy dephased state is also separable.

For three copies, use the global Majoranas and stabilizers
$Q_r=-\ii\gA_{2r}\gB_{2r-1}$ from
Eq.~\eqref{eq:Qr-definition}.  Define
\begin{align}\label{eq:three-copy-local-axes}
 A_0&=-\ii\gA_2\gA_4,
 & A_1&=-\ii\gA_4\gA_6,\notag\\
 C_0&=\ii\gB_1\gB_3,
 & C_1&=\ii\gB_3\gB_5.
\end{align}
These are even local Hermitian unitaries.  The CAR give
\begin{equation*}
 \{A_0,A_1\}=0,\qquad \{C_0,C_1\}=0,
 \qquad A_x^2=C_y^2=\id.
\end{equation*}
Alice measures $A_0$ or $A_1$, and Bob measures
\begin{equation}\label{eq:three-copy-bob-settings}
 B_0=\frac{C_0+C_1}{\sqrt2},
 \qquad
 B_1=\frac{C_0-C_1}{\sqrt2}.
\end{equation}
Anticommutation implies $B_0^2=B_1^2=\id$.
All four measurements preserve local parity, and Alice's observables
commute with Bob's.

Again using the CAR,
\begin{equation}\label{eq:three-copy-stabilizer-correlations}
 A_0C_0=Q_1Q_2,\qquad A_1C_1=Q_2Q_3.
\end{equation}
Equation~\eqref{eq:tau-m-product-projector} implies
$Q_r\tau^{\otimes3}=\tau^{\otimes3}$ for $r=1,2,3$.
The CHSH expression~\cite{Clauser1969} is therefore
\begin{align}\label{eq:three-copy-chsh-value}
 &\Tr\!\left[
 \tau^{\otimes3}\bigl(A_0(B_0+B_1)+A_1(B_0-B_1)\bigr)
 \right]\notag\\
 &\qquad=\sqrt2\,\Tr\!\left[
 \tau^{\otimes3}(Q_1Q_2+Q_2Q_3)\right]
 =2\sqrt2.
\end{align}
This is the quantum maximum and exceeds the local hidden-variable
bound $2$.

The measurements have Gaussian implementations.  Alice's observables
are Majorana bilinears, while Bob's can be written as
\begin{equation*}
 B_0=\ii\gB_3\frac{\gB_5-\gB_1}{\sqrt2},
 \qquad
 B_1=-\ii\gB_3\frac{\gB_1+\gB_5}{\sqrt2}.
\end{equation*}
Each is a bilinear in two orthonormal local Majoranas and can be
measured by a local Gaussian change of basis followed by an occupation
measurement.  Every outcome is retained.  For $m>3$, the same test
uses three copies and ignores the remaining ones.
\end{proof}

An indirect three-copy route can also be inferred from earlier
conversions.  Observation~4 of
Ref.~\cite{SpeeSchwaigerGiedkeKraus2018} converts two copies of $\tau$
into a pure fermionic Bell pair and local spectator modes.
Section~V.1 of Ref.~\cite{Debarba2020} describes how one such pair and
one copy of $\tau$ can be converted deterministically into an encoded
ebit.  The proposition gives a direct Gaussian measurement construction
and proves that fewer than three copies cannot violate a Bell inequality.

An alternative four-copy construction follows from the local Gaussian
conversion established in
Ref.~\cite[Observation~4]{SpeeSchwaigerGiedkeKraus2018} and written in
Eq.~\eqref{eq:pair-local-normal-form}.  Apply the conversion to two
disjoint pairs of copies of $\tau$ and discard the four spectator
modes.  The resulting state consists of two physical Bell pairs.
The use of two such pairs for Bell violation is already known
\cite{Dasenbrook2016,Debarba2020}.

In the Alice-first order $A_1,A_2,B_1,B_2$, the two retained Bell pairs
have 
the state vector
\begin{equation}\label{eq:four-copy-bell-output}
 \ket{\Omega}
 =\frac12\bigl(
 \ket{0,0,0,0}+\ket{0,1,0,1}+\ket{1,0,1,0}-\ket{1,1,1,1}
 \bigr).
\end{equation}
The minus sign is the fermionic crossing sign in
Eq.~\eqref{eq:general-copy-regrouping}.
Choose
\begin{align}\label{eq:four-copy-chsh-observables}
 \widehat A_0&=Z\otimes\id,
 &\widehat A_1&=X\otimes X,\notag\\
 \widehat B_0&=\frac{Z\otimes\id+Y\otimes Y}{\sqrt2},
 &\widehat B_1&=\frac{Z\otimes\id-Y\otimes Y}{\sqrt2}.
\end{align}
These observables square to the identity and commute with the
respective local parity $Z\otimes Z$.  The state $\ket{\Omega}$ has
\begin{equation*}
 \langle(Z\otimes\id)_A\otimes(Z\otimes\id)_B\rangle_\Omega=1,
 \qquad
 \langle(X\otimes X)_A\otimes(Y\otimes Y)_B\rangle_\Omega=1.
\end{equation*}
Hence
\begin{equation}\label{eq:four-copy-chsh-value}
 \bigl\langle
 \widehat A_0\otimes(\widehat B_0+\widehat B_1)
 +\widehat A_1\otimes(\widehat B_0-\widehat B_1)
 \bigr\rangle_\Omega
 =2\sqrt2.
\end{equation}
The conversion uses only local operations and is deterministic, with no
measurement outcome discarded.  Including this conversion
in the local measurement procedure gives a Bell violation for
$\tau^{\otimes4}$.
Proposition~\ref{prop:bell-activation-threshold} shows that three
copies already suffice and that this copy number is minimal.

Fermionic regrouping is also essential when taking reduced states.  An ordinary partial trace
of an interleaved subsystem may erase a parity string that belongs to the
retained mode algebra.  One must first adapt the ordering by fermionic
swaps, or equivalently restrict the state directly to the desired CAR
subalgebra.  All marginal states and entropies in 
this 
work
use that
restriction.

\section{Operational rates in fermionic Bell pairs}
\label{app:operational-rates}

\begin{corollary}[Formation and distillation in fermionic Bell pairs]
\label{cor:operational-tau}
Take $\beta:=\ketbra{\Phi_+}$ as the unit resource, with one
fermionic mode per party.  Allow local operations admitting parity-preserving dilations
with local physical ancillas, local discarding, and classical
communication (fermionic LOCC).  Define formation cost and
distillable entanglement using asymptotically vanishing trace
distance and count copies of $\beta$ per copy of the target.
Then
\begin{equation}\label{eq:operational-tau-rates}
 E_{\mathrm C}^{\beta}(\tau)
 =E_{\mathrm D}^{\beta}(\tau)=\frac12.
\end{equation}
The same rates hold when the protocols are restricted to Gaussian
local operations.  In entropy units, the value is $\tfrac12\log2$.
\end{corollary}

\begin{proof}
The known local Gaussian conversion of two copies
\cite[Observation~4]{SpeeSchwaigerGiedkeKraus2018}, written
explicitly in Eq.~\eqref{eq:pair-local-normal-form}, is reversible
after adjoining or discarding the two local maximally mixed
spectator modes.  Applying it to pairs gives Gaussian local
channels $\mathcal F_m$ and $\mathcal D_m$ satisfying
\begin{align}\label{eq:exact-formation-distillation-blocks}
 \mathcal F_m\!\left(\beta^{\otimes\lceil m/2\rceil}\right)
    &=\tau^{\otimes m},\\
 \mathcal D_m\!\left(\tau^{\otimes m}\right)
    &=\beta^{\otimes\lfloor m/2\rfloor}.
\end{align}
For odd $m$, formation first makes $m+1$ copies and discards the
last copy; distillation first discards the last input copy.
These exact protocols prove the two achievable rates.

For the converse, choose homogeneous Kraus operators for each local
instrument~\cite[Lemma~5]{SpeeSchwaigerGiedkeKraus2018}.
If $A_s:A\to A'$ has parity $p_s$, meaning
$P_{A'}A_s=(-1)^{p_s}A_sP_A$, then the Alice-first input and output
representations give the following Kraus matrices for the fermionic
extensions of Alice's and Bob's instruments, respectively,
\begin{equation}\label{eq:local-channel-product-kraus}
 \{A_s\otimes P_B^{p_s}\}_s,
 \qquad \{\id_A\otimes B_t\}_t.
\end{equation}
To verify the first formula, write an even local Stinespring isometry as
$V_A=\sum_s A_s\otimes\ket s_E$, where $\ket s_E$ has parity $p_s$.
The temporary ordering is $A',E,B$; moving $E$ past $B$ before tracing
it contributes precisely $P_B^{p_s}$.  For a Bob isometry the environment
is already last.  Every local physical ancillary state admits an
even-parity purification on additional local fermionic modes: pair its
definite-parity eigenvectors with orthogonal purifier vectors of matching
parity.  Such ancillas can therefore be absorbed into the dilation;
rectangular Kraus operators also cover changes in the local mode number.
Thus every complete LOCC branch is an ordinary
product matrix, with all reordering signs included.
Starting with $k$ Bell pairs, even after local
ancillas and discarding, the output has a pure-state decomposition
whose Schmidt ranks are at most $2^k$.  By Cauchy--Schwarz, its
overlap with a maximally entangled state of Schmidt rank $2^\ell$ is
at most $\min\{1,2^{k-\ell}\}$.  This bound is closed under
trace-norm limits, so it also applies to limits of finite-round LOCC
protocols.

Suppose a formation protocol makes $m$ copies from $k$ Bell pairs
with trace distance at most $\eps$, where trace distance means
$\tfrac12\norm{\rho-\sigma}_1$.  Compose it with $\mathcal D_m$.
Contractivity and the overlap bound give
\begin{equation}\label{eq:formation-schmidt-converse}
 1-\eps\leq\min\{1,2^{k-\lfloor m/2\rfloor}\}.
\end{equation}
For a distillation protocol making $\ell$ Bell pairs from $m$
copies with the same error, first apply $\mathcal F_m$ to
$\lceil m/2\rceil$ Bell pairs.  The bound becomes
\begin{equation}\label{eq:distillation-schmidt-converse}
 1-\eps\leq\min\{1,2^{\lceil m/2\rceil-\ell}\}.
\end{equation}
Vanishing error requires a formation rate at least $1/2$ and a
distillation rate at most $1/2$.  The argument applies to arbitrary
fermionic LOCC, completing both converses.
\end{proof}

This operational statement follows from the known two-copy
conversion.  It identifies the regularized roof with the formation
cost for this state and this Bell resource; no local parity
reference is assumed free.

\section{The exact Pfaffian convex roof}
\label{sec:pfaffian-roof}

The covariance matrix records only the two-point correlations of a fermionic
state.  As discussed in Sec.~\ref{sec:preliminaries}, a pure-Gaussian
decomposition of a Gaussian state must reproduce all even Majorana moments,
not merely the covariance matrix.  Wick's rule turns these conditions into a
finite family of Pfaffian constraints.

Let
\begin{equation}
\label{eq:even-index-set}
  \mathcal I_{2n}^{\ev,*}
  :=
  \set{
    I\subseteq\{1,\ldots,2n\}:
    |I|\ {\rm is\ even},\ I\neq\varnothing
  }.
\end{equation}
The number of nonempty even index sets is
\begin{equation}
  D
  :=
  \abs{\mathcal I_{2n}^{\ev,*}}
  =
  2^{2n-1}-1.
\end{equation}
For a covariance matrix \(L\), define its nontrivial even Pfaffian-moment
vector by
\begin{equation}
\label{eq:moment-map}
  m(L)
  :=
  \bigl(\Pf(L_I)\bigr)_
  {I\in\mathcal I_{2n}^{\ev,*}}
  \in\R^D.
\end{equation}

The decomposition constraint can be expressed as equality of
finite-dimensional moment vectors.

\begin{lemma}[Pfaffian characterization of Gaussian mixtures]
\label{lem:mixture-iff-pfaffian}
Let \(M\) be the covariance matrix of a fermionic Gaussian state, and let
\(K_\alpha\in\Jman_{2n}\).  For probabilities \(p_\alpha\), the following
statements are equivalent:
\begin{enumerate}[label=\textup{(\roman*)},leftmargin=2.2em]
\item
\begin{equation}
  \rho(M)
  =
  \sum_\alpha p_\alpha\rho(K_\alpha).
\end{equation}

\item For every even \(I\subseteq\{1,\ldots,2n\}\),
\begin{equation}
\label{eq:all-pfaffian-constraints}
  \Pf(M_I)
  =
  \sum_\alpha p_\alpha\Pf((K_\alpha)_I).
\end{equation}
\end{enumerate}
For \(I=\varnothing\), Eq.~\eqref{eq:all-pfaffian-constraints} is simply the
normalization condition \(\sum_\alpha p_\alpha=1\).
\end{lemma}

\begin{proof}
The Pfaffian expansion \eqref{eq:gaussian-pfaffian-expansion} gives
\begin{equation}
  \rho(M)-\sum_\alpha p_\alpha\rho(K_\alpha)
  =
  2^{-n}
  \sum_{I\ {\rm even}}
  \left[
    \Pf(M_I)
    -
    \sum_\alpha p_\alpha\Pf((K_\alpha)_I)
  \right]\Gamma_I.
\end{equation}
The even Majorana monomials \(\Gamma_I\) are linearly independent by
Eq.~\eqref{eq:majorana-orthogonality}.  The operator on the left therefore
vanishes if and only if every coefficient on the right vanishes.
\end{proof}

We can now express the Gaussian entanglement of formation as an ordinary
finite-dimensional constrained convex roof.

\begin{theorem}[Exact finite Pfaffian convex roof]
\label{thm:pfaffian-roof}
Let \(\rho(M)\) be a bipartite fermionic Gaussian state on \(n\) modes.  Then
\begin{equation}
\label{eq:pfaffian-roof}
  \EF^{\fG}(M)
  =
  \min_{\substack{
    N\in\mathbb N,\;
    p_\alpha\geq0,\;
    \sum_{\alpha=1}^N p_\alpha=1\\
    K_\alpha\in\Jman_{2n}}}
  \sum_{\alpha=1}^N p_\alpha E(K_\alpha),
\end{equation}
subject to
\begin{equation}
\label{eq:pfaffian-constraints}
  \sum_{\alpha=1}^N
  p_\alpha\Pf((K_\alpha)_I)
  =
  \Pf(M_I),
  \qquad
  I\in\mathcal I_{2n}^{\ev,*}.
\end{equation}
The minimum is attained, and an optimal ensemble can be chosen with
\begin{equation}
\label{eq:caratheodory-bound}
  N\leq D+1=2^{2n-1}.
\end{equation}
\end{theorem}

\begin{proof}
By Lemma~\ref{lem:mixture-iff-pfaffian}, the constraints
\eqref{eq:pfaffian-constraints}, together with normalization, are equivalent
to the density-operator constraint in Definition~\ref{def:fgeof}.  It remains
only to establish attainment and the finite ensemble bound.

Consider the compact set
\begin{equation}
  \mathcal G
  :=
  \set{
    \bigl(m(K),E(K)\bigr):
    K\in\Jman_{2n}
  }
  \subset\R^{D+1}.
\end{equation}
Compactness follows because \(\Jman_{2n}\) is compact and both \(m(K)\) and
\(E(K)\) are continuous.  Its convex hull
\(\mathcal C:=\conv(\mathcal G)\) is therefore compact.  Feasible Gaussian
ensembles correspond precisely to points of \(\mathcal C\) whose first
coordinate block is \(m(M)\).  This slice is nonempty because every Gaussian
state admits a decomposition into pure Gaussian occupation-number states; an
explicit decomposition is given in
Theorem~\ref{thm:canonical-decomposition}.  Compactness therefore guarantees
the existence of a point
\begin{equation}
  \bigl(m(M),\EF^{\fG}(M)\bigr)\in\mathcal C
\end{equation}
with minimal last coordinate.

Carathéodory's theorem first gives a representation of this point using at
most \(D+2\) elements of \(\mathcal G\).  Suppose this representation has \(D+2\)
positive weights.  The \(D+2\) vectors
\begin{equation}
  \bigl(1,m(K_\alpha)\bigr)\in\R^{D+1}
\end{equation}
are linearly dependent.  Hence there are real numbers \(a_\alpha\), not all
zero, such that
\begin{equation}
  \sum_\alpha a_\alpha=0,
  \qquad
  \sum_\alpha a_\alpha m(K_\alpha)=0.
\end{equation}
Changing the weights from \(p_\alpha\) to
\(p_\alpha+t a_\alpha\) therefore preserves normalization and every
Pfaffian constraint.  Choose the sign of \(a_\alpha\) so that
\begin{equation}
  \sum_\alpha a_\alpha E(K_\alpha)\leq0.
\end{equation}
For sufficiently small positive \(t\), all weights remain nonnegative.
Optimality forces the displayed derivative to vanish; otherwise the perturbed
ensemble would have strictly smaller average entanglement.  We may therefore
increase \(t\geq0\) until one weight becomes zero while preserving both
feasibility and optimality.
Removing the zero-weight branch leaves at most \(D+1\) states, proving
Eq.~\eqref{eq:caratheodory-bound}.
\end{proof}

For index sets with \(|I|=2\), the constraints
\eqref{eq:pfaffian-constraints} reduce to
\begin{equation}
  M=\sum_\alpha p_\alpha K_\alpha.
\end{equation}
These are only the covariance constraints.  The conditions with
\(|I|\geq4\) enforce equality of the higher even moments and are generally
independent of the covariance average.  Thus averaging pure covariance
matrices is necessary but not sufficient to produce a decomposition of the
corresponding Gaussian density operator.  Proposition~\ref{prop:covariance-not-enough}
gives a two-mode example of this distinction.
It is worth noting that the
number of Pfaffian moments grows exponentially with the number of modes.
Theorem~\ref{thm:pfaffian-roof} is thus a finite characterization of the
feasible ensembles, rather than an efficient evaluation algorithm.

\section{Why the bosonic covariance reduction fails}
\label{sec:additive-noise}

For a bosonic Gaussian state with covariance matrix \(V\), the condition
\(V_{\rm p}\leq V\) allows one to write \(V=V_{\rm p}+P\), with \(P\geq0\),
and to realize \(P\) by classical random phase-space displacements.  This
produces an ensemble of displaced copies of the pure Gaussian state with
covariance matrix \(V_{\rm p}\) and underlies the bosonic pure-covariance
reduction~\cite{WolfGiedkeKruegerWernerCirac2004}.  A direct fermionic
analogue fails for two independent reasons.  First, complete positivity
forbids a fermionic Gaussian channel that adds a fixed nonzero matrix to
every input covariance.  Second, even a formal convex average
of pure covariance matrices need not reproduce the corresponding Gaussian
density operator.

We begin with the channel obstruction.  A fermionic Gaussian channel acts
affinely on covariance matrices~\cite{Bravyi2005FLO,GreplovaGiedke2018}. Physical fermionic states have vanishing Majorana first moments,
so there is no parity-preserving analogue of the bosonic
displacement ensemble.

\begin{proposition}[No fermionic Gaussian additive-noise channel]
\label{prop:no-additive-channel}
Let \(\Phi\) be a fermionic Gaussian channel with covariance action
\begin{equation}
\label{eq:fgc-action}
  M\longmapsto XMX^\transpose+N,
\end{equation}
where \(X\) is real and \(N\) is real antisymmetric.  In the covariance
convention \eqref{eq:cm-definition}, complete positivity is equivalent to
\begin{equation}
\label{eq:fgc-cp}
  \id+\ii N-XX^\transpose\geq0.
\end{equation}
If the input and output Majorana spaces have the same dimension and
\(X=\id\), then \(N=0\).  Hence no fermionic Gaussian channel has the
nontrivial additive-noise action
\begin{equation}
  M\longmapsto M+N,
  \qquad N\neq0.
\end{equation}
\end{proposition}

\begin{proof}
Setting \(X=\id\) in Eq.~\eqref{eq:fgc-cp} gives
\begin{equation}
  \ii N\geq0.
\end{equation}
Because \(N\) is real and antisymmetric, \(\ii N\) is Hermitian and has trace
zero.  A positive semidefinite Hermitian matrix with zero trace must vanish.
Therefore \(N=0\).  Independently of the representation by \(X\),
a channel acting as \(M\mapsto M+N\) could be iterated on the
maximally mixed state to produce covariance \(kN\) for every integer
\(k\).  The bounded covariance singular values then also force \(N=0\).
\end{proof}

Proposition~\ref{prop:no-additive-channel} rules out the physical channel
that would play the role of bosonic random-displacement noise.  By itself,
however, it does not exclude the possibility that a suitable ensemble of
pure fermionic Gaussian states might reproduce the target state without
arising from such a channel.  The following two-mode example shows why
covariance equality alone does not guarantee such a decomposition.

\begin{proposition}[Covariance averaging is not state averaging]
\label{prop:covariance-not-enough}
Let
\begin{equation}
  J:=
  \begin{pmatrix}
    0&1\\
    -1&0
  \end{pmatrix},
  \qquad
  K_+:=J\oplus J,
  \qquad
  K_-:=-J\oplus(-J).
\end{equation}
Then \(K_+,K_-\in\Jman_4\) and
\begin{equation}
  \frac12K_++\frac12K_-=0.
\end{equation}
Nevertheless,
\begin{equation}
\label{eq:state-average-not-maxmixed}
  \frac12\rho(K_+)+\frac12\rho(K_-)
  \neq
  \rho(0)
  =
  \frac14\id.
\end{equation}
\end{proposition}

\begin{proof}
The matrices \(K_\pm\) satisfy \(K_\pm^2=-\id\), so they are pure Gaussian
covariance matrices.  Their average vanishes by construction.

To distinguish the density operators, consider the fourth-order Majorana
monomial
\begin{equation}
  \Gamma_{1,2,3,4}:=\ii^2\gamma_1\gamma_2\gamma_3\gamma_4.
\end{equation}
For two modes this operator is the total parity.  Wick's rule gives
\begin{equation}
  \Tr\!\left[\rho(K_\pm)\Gamma_{1,2,3,4}\right]
  =
  \Pf(K_\pm).
\end{equation}
Using multiplicativity of the Pfaffian under direct sums,
\begin{equation}
  \Pf(K_+)=\Pf(J)^2=1,
  \qquad
  \Pf(K_-)=\Pf(-J)^2=1.
\end{equation}
Consequently,
\begin{equation}
  \Tr\!\left[
    \left(
      \frac12\rho(K_+)+\frac12\rho(K_-)
    \right)\Gamma_{1,2,3,4}
  \right]
  =1.
\end{equation}
On the other hand,
\begin{equation}
  \Tr\!\left[\rho(0)\Gamma_{1,2,3,4}\right]
  =
  \Pf(0)
  =
  0.
\end{equation}
The two states therefore differ in their fourth-order correlation, despite
having the same covariance matrix.
\end{proof}

The two propositions identify distinct failures of the bosonic argument.
Proposition~\ref{prop:no-additive-channel} shows that the required
additive-noise transformation is not a physical fermionic Gaussian channel,
whereas Proposition~\ref{prop:covariance-not-enough} shows that covariance
averaging would remain insufficient even if considered only as a formal
ensemble construction.  The missing information is precisely the collection
of higher Pfaffian moments appearing in
Theorem~\ref{thm:pfaffian-roof}.

\subsection{Optimal decompositions need not be flat}
\label{sec:flatness-fails}

The bosonic pure-covariance construction has another important consequence:
all pure states in the displacement ensemble have the same entanglement.  We
call such an ensemble \emph{entanglement-flat}.  The fermionic convex roof
does not share this property, even for one mode held by each party.

\begin{theorem}[A non-flat optimal decomposition]
\label{thm:flatness-counterexample}
Consider two fermionic modes, one held by Alice and one by Bob.  For
\(0<q<1\), let
\begin{equation}
\label{eq:rho-q}
  \rho_q
  :=
  (1-q)\ketbra{0,0}
  +
  q\ketbra{\Psi_+},
  \qquad
  \ket{\Psi_+}
  :=
  \frac{\ket{1,0}+\ket{0,1}}{\sqrt2}.
\end{equation}
Then \(\rho_q\) is fermionic Gaussian and
\begin{equation}
\label{eq:rho-q-fgeof}
  \EF^{\mathrm{phys}}(\rho_q)
  =
  \EF^{\fG}(\rho_q)
  =
  q\log 2.
\end{equation}
Moreover, every physical pure-state decomposition of \(\rho_q\), and hence
every pure-Gaussian decomposition, consists---up to phases and splitting
identical branches---of \(\ket{0,0}\) and \(\ket{\Psi_+}\) with total weights
\(1-q\) and \(q\), respectively.  Consequently, no optimal Gaussian
decomposition of \(\rho_q\) is entanglement-flat.
\end{theorem}

\begin{proof}
Introduce the delocalized modes
\begin{equation}
  f_+
  :=
  \frac{c_A+c_B}{\sqrt2},
  \qquad
  f_-
  :=
  \frac{c_A-c_B}{\sqrt2}.
\end{equation}
Since \(\ket{\Psi_+}=f_+^\dagger\ket{0,0}\), the state factorizes in this
normal-mode basis as
\begin{equation}
\label{eq:rho-q-delocalized}
  \rho_q
  =
  \left[
    (1-q)\ketbra{0_+}
    +
    q\ketbra{1_+}
  \right]
  \otimes
  \ketbra{0_-}.
\end{equation}
Every parity-preserving one-mode state is Gaussian, products of Gaussian
states are Gaussian, and the transformation
\((c_A,c_B)\mapsto(f_+,f_-)\) is a Gaussian mode rotation.  Hence
\(\rho_q\) is fermionic Gaussian.

Now consider an arbitrary physical pure-state decomposition of \(\rho_q\).
Every branch has definite total parity.  The even block of \(\rho_q\) has
support \(\Span\{\ket{0,0}\}\), while its odd block has support
\(\Span\{\ket{\Psi_+}\}\).  Lemma~\ref{lem:support} therefore forces every
even branch to equal \(\ket{0,0}\) up to phase and every odd branch to equal
\(\ket{\Psi_+}\) up to phase.  The traces of the two parity blocks fix their
total weights to \(1-q\) and \(q\).
Both branches are pure Gaussian, and their entanglements are
\begin{equation}
  E(\ket{0,0})=0,
  \qquad
  E(\ket{\Psi_+})=\log 2.
\end{equation}
Thus every physical or Gaussian decomposition has average entanglement
\(q\log2\), proving Eq.~\eqref{eq:rho-q-fgeof}.  Because \(0<q<1\), both
branches occur with positive weight and have unequal entanglement.
\end{proof}

The rank-two formula in Theorem~\ref{thm:single-defect-formula}
extends this support argument to arbitrary numbers of modes.

\section{Exact formation for rank-two Gaussian states}
\label{sec:single-defect-formula}

A rank-two Gaussian state has one mixed normal mode and a pure
occupation state on the remaining normal modes.  Its two spectral
branches have opposite parity, forcing every physical decomposition to
use these same two pure states.  This gives the following formula.

\begin{definition}[Rank-two Gaussian state]
\label{def:single-defect}
A fermionic Gaussian state \(\rho(M)\) has density-operator rank two if and only if
exactly one block in its canonical form has modulus strictly smaller than
one.  Equivalently, after relabelling the normal modes,
\begin{equation}
\label{eq:single-defect-canonical-form}
  M
  =
  O\left(
    \mu J
    \oplus
    \bigoplus_{j=2}^n\sigma_jJ
  \right)O^\transpose,
  \qquad
  -1<\mu<1,
  \qquad
  \sigma_j\in\{\pm1\},
\end{equation}
for some \(O\in SO(2n)\).  The first block is the unique mixed normal mode,
while all remaining normal modes are pure.
\end{definition}

\begin{theorem}[Exact rank-two formula]
\label{thm:single-defect-formula}
Let \(\rho(M)\) be a rank-two Gaussian state with canonical form
\eqref{eq:single-defect-canonical-form}.  Define its two pure Gaussian
completions by
\begin{equation}
\label{eq:single-defect-Kpm}
  K_+
  :=
  O\left(
    J
    \oplus
    \bigoplus_{j=2}^n\sigma_jJ
  \right)O^\transpose,
  \qquad
  K_-
  :=
  O\left(
    -J
    \oplus
    \bigoplus_{j=2}^n\sigma_jJ
  \right)O^\transpose,
\end{equation}
and let
\begin{equation}
\label{eq:single-defect-pm}
  p_+
  :=
  \frac{1+\mu}{2},
  \qquad
  p_-
  :=
  \frac{1-\mu}{2}.
\end{equation}
Then
\begin{equation}
\label{eq:single-defect-state-decomp}
  \rho(M)
  =
  p_+\rho(K_+)+p_-\rho(K_-).
\end{equation}
Moreover, every physical pure-state decomposition of \(\rho(M)\), and hence
every pure-Gaussian decomposition, consists---up to phases and splitting
identical branches---of these same two states with total weights \(p_+\) and
\(p_-\).  Consequently,
\begin{equation}
\label{eq:single-defect-eof}
  \EF^{\mathrm{phys}}(\rho(M))
  =
  \EF^{\fG}(\rho(M))
  =
  p_+E(K_+)+p_-E(K_-).
\end{equation}
Here \(E(K_\pm)\) is evaluated with respect to the original Alice--Bob mode
partition.
\end{theorem}

\begin{proof}
Equation~\eqref{eq:single-defect-state-decomp} is the \(k=1\) case of the
canonical spectral decomposition in
Theorem~\ref{thm:canonical-decomposition}.  Let
\(\ket{\Omega_+}\) and \(\ket{\Omega_-}\) denote the pure Gaussian vectors
with covariance matrices \(K_+\) and \(K_-\), respectively.

In the canonical normal-mode basis, these vectors differ only in the
occupation of the mixed normal mode.  They are therefore orthogonal and have
opposite total parity.  The Gaussian unitary implementing
\(O\in SO(2n)\) preserves total parity, so the same statements hold in the
original mode basis.  Hence
\begin{equation}
  \supp\rho(M)
  =
  \Span\{\ket{\Omega_+},\ket{\Omega_-}\},
\end{equation}
with one support vector in each parity sector.

Consider an arbitrary physical pure-state decomposition
\begin{equation}
  \rho(M)
  =
  \sum_\alpha q_\alpha\ketbra{\psi_\alpha},
  \qquad
  q_\alpha>0.
\end{equation}
By Lemma~\ref{lem:support}, every \(\ket{\psi_\alpha}\) lies in
\(\supp\rho(M)\).  Each branch also has definite total parity by the
definition of the physical convex roof.  Because the support contains only
one line of each parity, every branch must equal either
\(\ket{\Omega_+}\) or \(\ket{\Omega_-}\), up to phase.  Projecting the
decomposition onto these two orthogonal lines fixes their total weights to
\(p_+\) and \(p_-\).

Thus every physical decomposition has average entanglement
\begin{equation}
  p_+E(K_+)+p_-E(K_-).
\end{equation}
Both branches are Gaussian, so this decomposition is feasible for both
\(\EF^{\mathrm{phys}}\) and \(\EF^{\fG}\), proving
Eq.~\eqref{eq:single-defect-eof}.
\end{proof}

Although Theorem~\ref{thm:single-defect-formula} is written using a canonical
normal-mode basis, its content is basis-independent: the covariance matrix
\(M\) determines a unique two-dimensional mixed normal-mode plane, and
\(K_\pm\) are its two pure Gaussian completions.  Thus evaluating the formula
requires only the standard orthogonal normal-form decomposition of \(M\).

\begin{corollary}[Non-flatness for rank-two Gaussian states]
\label{cor:single-defect-nonflat}
Let \(\rho(M)\) have rank two.  If its two pure completions satisfy
\begin{equation}
  E(K_+)\neq E(K_-),
\end{equation}
then no optimal physical or Gaussian decomposition of \(\rho(M)\) is
entanglement-flat.
\end{corollary}

\begin{proof}
Both branches occur with positive weight because \(-1<\mu<1\), and
Theorem~\ref{thm:single-defect-formula} shows that every admissible
decomposition contains only these two branches.  If their entanglements are
unequal, such a decomposition cannot be flat.
\end{proof}

The two-mode family \(\rho_q\) from Eq.~\eqref{eq:rho-q} is the simplest
instance of the rank-two formula.

\section{When the Gaussian restriction is automatic}
\label{sec:physical-vs-gaussian}

The physical convex roof permits all pure states of definite parity,
whereas the Gaussian convex roof restricts the branches to pure
Gaussian states.  For Gaussian targets of sufficiently low rank,
however, the support of the state makes this restriction automatic.

\begin{lemma}[Pure states on at most three modes]
\label{lem:pure-spinor}
Every pure state of definite parity on \(k\le3\) fermionic modes is
Gaussian~\cite[Lemma~3 and Sec.~III.E]{Bravyi2005Capacity}; see also
Ref.~\cite[Proposition~1]{DeMeloCwiklinskiTerhal2013}.
\end{lemma}

\begin{proof}
Bravyi's statement covers even-parity vectors; the odd-parity case follows
from a Majorana conjugation, as explained below.  We give a direct proof.
The claim is immediate for one mode and follows from
Lemma~\ref{lem:two-mode-pure-gaussian} for two modes.  For three modes,
an even-parity state has the form
\begin{equation}
  \ket{\psi}
  =
  \alpha\ket{000}
  +\sum_{1\le i<j\le3}z_{ij}c_i^\dagger c_j^\dagger\ket{000}.
\end{equation}
If \(\alpha\ne0\), then
\begin{equation}
  \ket{\psi}
  =
  \alpha
  \exp\!\left(
    \alpha^{-1}
    \sum_{i<j}z_{ij}c_i^\dagger c_j^\dagger
  \right)\ket{000},
\end{equation}
because the square of the pair-creation operator vanishes on three
modes.  This is a BCS Gaussian state.  If \(\alpha=0\), the state is a
two-particle state on three modes.  Every two-particle state on a
three-dimensional one-particle space is a single Slater determinant,
and is therefore Gaussian.

Finally, multiplication by a Majorana operator exchanges the two
parity sectors and preserves Gaussianity: conjugation by a Majorana
sends each Majorana to itself or its negative and preserves Wick's rule.
The odd-parity case therefore follows from the even-parity case.
\end{proof}

\begin{theorem}[Low-rank equivalence of physical and Gaussian roofs]
\label{thm:low-rank-equality}
Let \(\rho(M)\) be an \(n\)-mode fermionic Gaussian state, and write its
canonical form as
\begin{equation}
  M
  =
  O\left(\bigoplus_{j=1}^{n}\nu_jJ\right)O^\transpose.
\end{equation}
Let
\begin{equation}
  k:=\#\{j:|\nu_j|<1\}
\end{equation}
be the number of mixed normal modes, so that
\(\rank\rho(M)=2^k\).  If \(k\le3\), every physical pure-state
decomposition of \(\rho(M)\) consists entirely of Gaussian branches.
Consequently,
\begin{equation}\label{eq:low-rank-equality}
  \EF^{\mathrm{phys}}(\rho(M))
  =
  \EF^{\fG}(\rho(M)).
\end{equation}
\end{theorem}

\begin{proof}
To prove automatic Gaussianity of the ensemble branches, conjugate by the
Gaussian unitary implementing \(O^\transpose\).  This preserves parity
and Gaussianity, so it suffices to consider the
canonical state
\begin{equation}
  \rho(D)
  =
  \rho_{\mathrm{mix}}\otimes\ketbra{\chi},
  \qquad
  \rho_{\mathrm{mix}}
  =
  \bigotimes_{j=1}^{k}
  \frac12(\id+\nu_jQ_j),
\end{equation}
where \(\ket{\chi}\) is the pure occupation state of the remaining
normal modes.  Since \(|\nu_j|<1\) for \(j\le k\),
\(\rho_{\mathrm{mix}}\) is full rank on the Fock space
\(\mathcal H_{\mathrm{mix}}\) of the \(k\) mixed modes.  Hence
\begin{equation}
  \supp\rho(D)
  =
  \mathcal H_{\mathrm{mix}}\otimes\C\ket{\chi}.
\end{equation}

By Lemma~\ref{lem:support}, every branch in a pure-state decomposition
of \(\rho(D)\) must therefore have the form
\begin{equation}
  \ket{\varphi}\otimes\ket{\chi}.
\end{equation}
Because \(\ket{\chi}\) has definite parity, definiteness of the total
parity forces \(\ket{\varphi}\) to have definite parity on the \(k\)
mixed modes.  For \(k\le3\), Lemma~\ref{lem:pure-spinor} implies that
\(\ket{\varphi}\) is Gaussian.  Its product with \(\ket{\chi}\) is
Gaussian, and conjugating back by the Gaussian unitary shows that the
original branch is Gaussian as well.

Thus the physical and Gaussian convex roofs have exactly the same
feasible ensembles, which proves \eqref{eq:low-rank-equality}.
\end{proof}

\begin{remark}[Sharpness of automatic Gaussianity]
The condition \(k\le3\) is optimal for the statement that every physical
ensemble branch is Gaussian.  The four-mode state vector
\begin{equation}\label{eq:ghz-four}
  \ket{\psi_{\mathrm{GHZ}}}
  :=
  \frac{\ket{0,0,0,0}+\ket{1,1,1,1}}{\sqrt2}
\end{equation}
has even parity, but all of its off-diagonal two-point Majorana correlations
vanish.  Its covariance matrix is therefore zero, and the pure-state
condition \(M^2=-\id\) shows that it is not Gaussian.

If a Gaussian state has \(k\ge4\) mixed normal modes, then in the
canonical basis its support contains
\begin{equation}
  \ket{\psi_{\mathrm{GHZ}}}
  \otimes
  \ket{0}^{\otimes(k-4)}
  \otimes
  \ket{\chi}.
\end{equation}
Because the Gaussian state is strictly positive on its support, this
vector can be assigned a sufficiently small positive weight in a
physical pure-state decomposition, with the positive remainder
decomposed into parity-definite eigenvectors.  Thus non-Gaussian
physical branches become available when \(k\ge4\).  This makes the
physical feasible set strictly larger, although it does not by itself
imply that the value of the physical convex roof is strictly smaller.
\end{remark}

\section{Comparison with fermionic logarithmic negativity}
\label{sec:negativity-comparison}

Shapourian and Ryu~\cite[Eq.~(88)]{ShapourianRyu2019} evaluated the
fermionic logarithmic negativity \cite{PhD,VidalNegativity,PlenioNegativity} 
of the same state $\tau$ and obtained
$\log\sqrt2$.  We record the two-mode formula in our conventions and
compare this known result with the exact regularized formation value.
This comparison is independent of the proof of
Theorem~\ref{thm:many-copy-nonadditivity}.

Denote the fermionic partial transpose on Alice by $T_A^{\mathrm f}$ and define
\begin{equation}
  \mathcal E_N^{\mathrm f}(\rho):=
  \log\left\|\rho^{T_A^{\mathrm f}}\right\|_1 .
\end{equation}
For a parity-preserving two-mode state written as in \eqref{eq:two-mode-block-state},
the fermionic partial transpose is, up to the local particle-hole unitary appearing
in the occupation-basis form of Ref.~\cite{ShapourianRyu2019}, represented in the
same block-ordered basis by
\begin{equation}
  \widetilde\rho^{T_A^{\mathrm f}}
  =
  \begin{pmatrix}
    a&-\ii w&0&0\\
    -\ii w^*&d&0&0\\
    0&0&b&-\ii z\\
    0&0&-\ii z^*&c
  \end{pmatrix} .
\end{equation}
The unitary does not change the trace norm.  Therefore
\begin{equation}\label{eq:two-mode-fermionic-negativity-norm}
  \left\|\rho^{T_A^{\mathrm f}}\right\|_1
  =
  \sqrt{(a+d)^2+4|w|^2}
  +
  \sqrt{(b+c)^2+4|z|^2}.
\end{equation}
Equivalently, in the notation of Appendix~\ref{sec:two-mode-formula},
\begin{equation}\label{eq:two-mode-fermionic-log-negativity}
  \mathcal E_N^{\mathrm f}(\rho)
  =
  \log\!\left(
  \sqrt{p_e^2+4|w|^2}
  +
  \sqrt{p_o^2+4|z|^2}
  \right).
\end{equation}
Indeed, a block $B$ with diagonal entries $x,y\ge0$ and off-diagonal
entries $-\ii u,-\ii u^*$ satisfies
\[
 \|B\|_1^2=\Tr(B^\dagger B)+2|\det B|
 =x^2+y^2+2|u|^2+2(xy+|u|^2)
 =(x+y)^2+4|u|^2.
\]

For the state $\tau$ in \eqref{eq:tau-definition}, one has
\begin{equation}
  p_e=p_o=\frac12,
  \qquad
  |z|=|w|=\frac14.
\end{equation}
Thus \eqref{eq:two-mode-fermionic-log-negativity} gives
\begin{equation}\label{eq:tau-fermionic-log-negativity}
  \mathcal E_N^{\mathrm f}(\tau)
  =\log\sqrt2
  =\frac12\log2.
\end{equation}
The fermionic logarithmic negativity is additive under tensor products in this
convention~\cite{ShapourianRyu2019}, hence
\begin{equation}
  \mathcal E_N^{\mathrm f}(\tau^{\otimes m})
  =\frac m2\log2.
\end{equation}
Combining this with Corollary~\ref{cor:regularized-tau} yields
\begin{equation}
  \lim_{m\to\infty}\frac1m \EF^{\mathrm{phys}}(\tau^{\otimes m})
  =\lim_{m\to\infty}\frac1m \EF^{\fG}(\tau^{\otimes m})
  =\mathcal E_N^{\mathrm f}(\tau).
\end{equation}
Thus, for this particular Gaussian state, the one-copy formation value is twice
the fermionic logarithmic negativity, while the regularized fermionic Gaussian
formation value coincides with it.  This equality is special to the example and
is not claimed as a general identity between $\EF^{\fG}$ and fermionic logarithmic
negativity.

\section{Additivity of the physical-extension squashed quantity}
\label{app:squashed-entanglement}

For comparison, the extension-based quantity
\begin{equation}\label{eq:physical-squashed-definition}
 E_{\mathrm{sq}}^{\mathrm{phys}}(A:B)_\rho
 :=\frac12\inf_{\substack{\omega_{A,B,E}\,\mathrm{physical}\\
                         \omega_{A,B}=\rho}}
 I(A:B\mid E)_\omega
\end{equation}
is additive under fermionic 
products.  The infimum allows any finite
number of extension modes and imposes no Gaussian restriction.
All entropies in
$I(A:B\mid E)=S(A,E)+S(B,E)-S(E)-S(A,B,E)$ are entropies of fermionic
restrictions to the indicated mode sets.
The proof is the conditional-mutual-information argument of
Ref.~\cite{ChristandlWinter2004}, using strong subadditivity for CAR
systems; see also the discussion in Ref.~\cite[Sec.~8]{Szalay2021}.
The argument below keeps track of the physicality of each extension.

Let $\omega_{A_1,A_2,B_1,B_2,E}$ be a physical extension of
$\rho_{A_1,B_1}\otimes\sigma_{A_2,B_2}$.
The chain rule and nonnegativity of conditional mutual information give
\begin{align}
 I(A_1,A_2:B_1,B_2\mid E)_\omega
 &=I(A_1:B_1\mid E)_\omega
   +I(A_1:B_2\mid E,B_1)_\omega\nonumber\\
 &\quad+I(A_2:B_1\mid E,A_1)_\omega
   +I(A_2:B_2\mid E,A_1,B_1)_\omega\nonumber\\
 &\geq I(A_1:B_1\mid E)_\omega
       +I(A_2:B_2\mid E,A_1,B_1)_\omega\nonumber\\
 &\geq2E_{\mathrm{sq}}^{\mathrm{phys}}(\rho)
       +2E_{\mathrm{sq}}^{\mathrm{phys}}(\sigma).
 \label{eq:physical-squashed-lower-bound}
\end{align}
In the last line, $\omega_{A_1,B_1,E}$ is an admissible extension of
$\rho$, and the entire state, with $E,A_1,B_1$ as extension system, is an
admissible extension of $\sigma$.  Fermionic restrictions of physical
states are physical, so both applications respect
Eq.~\eqref{eq:physical-squashed-definition}.
Taking the infimum gives the lower bound for additivity.

For the reverse bound, choose physical extensions of $\rho$ and $\sigma$
with conditional mutual informations arbitrarily close to their
respective infima.  Their fermionic product is a physical extension of
$\rho\otimes\sigma$, and entropy additivity makes its conditional
mutual information the sum of the two original values.  Hence
\begin{equation}\label{eq:physical-squashed-additivity}
 E_{\mathrm{sq}}^{\mathrm{phys}}(\rho\otimes\sigma)
 =E_{\mathrm{sq}}^{\mathrm{phys}}(\rho)
  +E_{\mathrm{sq}}^{\mathrm{phys}}(\sigma).
\end{equation}
This comparison is independent of the formation and communication
proofs in the main text.

\section{Proof of Theorem \ref{thm:quantum-capacity}}
\label{Proof}

In this section, we 
prove 
Theorem \ref{thm:quantum-capacity}.

\begin{proof}
This proof will require some preparation. 
For what follows, write $r_j:=\gamma_{2j-1}$ for the $m$ retained Majoranas, and let
$\mathcal A_m$ be the complex algebra spanned by their even monomials,
including the identity.  Equation~\eqref{eq:channel-monomial-action}
and physicality show that every output belongs to $\mathcal A_m$ and
that every element of this algebra is fixed by $\Phi_m^*$.
The even complex Clifford algebra has the representation
\begin{equation}\label{eq:retained-even-algebra}
 \mathcal A_m\cong
 \begin{cases}
  \C,&m=1,\\
  M_{2^q}(\C),&m=2q+1,\quad q\geq1,\\
  M_{2^{q-1}}(\C)\oplus M_{2^{q-1}}(\C),
      &m=2q,\quad q\geq1.
 \end{cases}
\end{equation}
Here $M_d(\C)$ denotes the algebra of $d\times d$ matrices.
For odd $m$, the $2q$ operators $\ii r_jr_m$, $j<m$, are
anticommuting Hermitian unitaries generating $M_{2^q}(\C)$.
For even $m$, the central Hermitian unitary
$\ii^q r_1\cdots r_{2q}$ splits the even algebra into two sectors;
in each, $\ii r_jr_{2q-1}$, $1\leq j\leq2q-2$, generate
$M_{2^{q-1}}(\C)$.  Independence of the Majorana monomials gives
the displayed algebra dimensions and excludes additional relations.

In its representation on the output Hilbert space, this algebra acts as
\begin{equation}\label{eq:logical-algebra-multiplicity}
 \bigoplus_s\bigl(M_{d_s}(\C)\otimes\id_{g_s}\bigr),
 \qquad d_s\leq2^{q_m},
\end{equation}
with multiplicities $g_s$.  Hence, including an arbitrary reference,
every encoded channel output has the form
\begin{equation}\label{eq:reference-logical-output}
 \omega_{R,B}=\bigoplus_s p_s\,\omega_{R,L_s}
                  \otimes\frac{\id_{g_s}}{g_s},
 \qquad \dim L_s=d_s.
\end{equation}
The reference is encoded at fixed parity, so the fermionic extension
acts as the identity on that logical reference.  The assertion follows
also by expanding each reference matrix coefficient in physical input
operators and applying Eq.~\eqref{eq:channel-monomial-action}.

Every pure vector on $R L_s$ has Schmidt rank at most $d_s$.
Decomposing the states in Eq.~\eqref{eq:reference-logical-output}
and the multiplicity factors therefore shows that $\omega_{R,B}$ has
Schmidt number at most $2^{q_m}$, meaning that it admits a 
pure-state
decomposition with this Schmidt-rank bound.  A receiver's local
decoding cannot increase 
Schmidt number. By the Cauchy--Schwarz inequality, every normalized vector
$\ket{\psi}$ of Schmidt rank at most $d$ satisfies
\[
 \bigl|\braket{\Omega_D}{\psi}\bigr|^2
 \leq \min\{1,d/D\}.
\]
Averaging over a pure-state decomposition gives the same bound
for states of Schmidt number at most $d$. Applying this bound to
the decoded state proves Eq.~\eqref{eq:quantum-fidelity-converse}.
Exact transmission therefore requires $D\leq2^{q_m}$.

For achievability, first let $m=2q+1$.  With $P$ the input parity,
take the simultaneous $+1$ eigenspace of the commuting observables
\begin{equation}\label{eq:quantum-code-stabilizers}
 P,\qquad H_j:=\ii\gamma_{4j-2}\gamma_{4j},\quad 1\leq j\leq q.
\end{equation}
Their $q+1$ independent constraints give a code subspace
\begin{equation}\label{eq:quantum-code-projector}
 \Pi_{\mathrm{code}}=
   \frac{\id+P}{2}\prod_{j=1}^q\frac{\id+H_j}{2},
 \qquad \rank\Pi_{\mathrm{code}}=2^q.
\end{equation}
Independence follows because every nonidentity product of these
stabilizers is a nonidentity Majorana monomial and has zero trace.
Every $H_j$ contains only erased Majoranas and commutes with
$\mathcal A_m$, as does $P$.  Thus the code is invariant under
$\mathcal A_m$.  Its dimension equals the irreducible representation
dimension $2^q$, so the restricted algebra is the full logical matrix
algebra.  All its observables are fixed by the adjoint channel.
Encoding a logical state in this subspace and tracing out the output
multiplicity subsystem consequently recover the state exactly,
including its correlations with any reference.

This decoder is compatible with parity superselection.  Since output
parity commutes with $\mathcal A_m$, it acts only on the multiplicity
subsystem.  Choose a multiplicity basis of parity eigenvectors and embed
the recovered logical register in an even-parity subspace of $q+1$
output modes.  Each decoding Kraus operator, obtained by contracting
with one multiplicity basis vector, then has definite parity.  Such
operators admit an even dilation by recording their parity in a local
environment, as in Appendix~\ref{app:operational-rates}.  The encoder
likewise maps an even-parity logical register into the even code
subspace.  The construction makes no Gaussian restriction on these
encoding and decoding operations.

For even $m=2q$, use the odd-block code on the first $m-1$ modes,
append a vacuum mode, and discard its output.  This transmits
$q-1=q_m$ logical qubits exactly.  The trivial one-dimensional code
also covers $m=1$. Since $q_m/m\to1/2$, these exact codes achieve the rate in
Eq.~\eqref{eq:quantum-capacity}. Conversely, if
$D_m\geq2^{m(1/2+\delta)}$ for some $\delta>0$, then
Eq.~\eqref{eq:quantum-fidelity-converse} gives
$F_{\mathrm e}\leq2^{-m\delta}$. This proves the strong converse.
\end{proof}

For three uses, the retained bilinears
$X_L=\ii\gamma_1\gamma_3$ and $Z_L=\ii\gamma_3\gamma_5$
are parity-even logical Pauli operators: they square to the identity,
anticommute, and are fixed together with their product.  Thus three
uses transmit one logical qubit perfectly, whereas one or two uses
transmit none.  At even blocklength $m=2q$, the two central sectors in
Eq.~\eqref{eq:retained-even-algebra} also explain the classical optimum:
one classical sector label accompanies a register of $q-1$ qubits.
The asymptotic capacities follow from this retained algebra and the
fermionic composition, rather than ordinary qubit tensor powers.

\bibliography{fermionic_gaussian}

@article{ChitambarGour2019,
  author  = {Chitambar, E. and Gour, G.},
  title   = {Quantum resource theories},
  journal = {Rev. Mod. Phys.},
  volume  = {91},
  number  = {2},
  pages   = {025001},
  year    = {2019},
  doi     = {10.1103/RevModPhys.91.025001}
}

@article{EisertScheelPlenio2002,
  author  = {Eisert, J. and Scheel, S. and Plenio, M. B.},
  title   = {Distilling {Gaussian} states with {Gaussian} operations
             is impossible},
  journal = {Phys. Rev. Lett.},
  volume  = {89},
  pages   = {137903},
  year    = {2002},
  doi     = {10.1103/PhysRevLett.89.137903}
}

@article{Fiurasek2002,
  author  = {Fiur{\'a}{\v{s}}ek, J.},
  title   = {{Gaussian} transformations and distillation of entangled
             {Gaussian} states},
  journal = {Phys. Rev. Lett.},
  volume  = {89},
  pages   = {137904},
  year    = {2002},
  doi     = {10.1103/PhysRevLett.89.137904}
}

@article{GiedkeCirac2002,
  author  = {Giedke, G. and Cirac, J. I.},
  title   = {Characterization of {Gaussian} operations and distillation
             of {Gaussian} states},
  journal = {Phys. Rev. A},
  volume  = {66},
  pages   = {032316},
  year    = {2002},
  doi     = {10.1103/PhysRevA.66.032316}
}

@misc{LiuEtAlPurification2025,
  author        = {Liu, Zhenhuan and Du, Zhenyu and Eisert, Jens
                   and Cai, Zhenyu and Liu, Zi-Wen},
  title         = {No universal purification in quantum mechanics},
  year          = {2025},
  eprint        = {2509.21111},
  archivePrefix = {arXiv} 
}

@article{PlenioVirmani2007,
  author  = {Plenio, M. B. and Virmani, S.},
  title   = {An introduction to entanglement measures},
  journal = {Quant. Inf. Comp.},
  volume  = {7},
  number  = {1--2},
  pages   = {1--51},
  year    = {2007},
  doi     = {10.26421/QIC7.1-2-1}
}

@article{Araki1970,
  author = {Araki, Huzihiro},
  title = {On Quasifree States of {CAR} and {Bogoliubov} Automorphisms},
  journal = {Publications of the Research Institute for Mathematical Sciences},
  volume = {6},
  number = {3},
  pages = {385--442},
  year = {1971},
  doi = {10.2977/prims/1195193913}
}

@article{Bravyi2005FLO,
  author = {Bravyi, Sergey},
  title = {Lagrangian representation for fermionic linear optics},
  journal = {Quant. Inf. Comp.},
  volume = {5},
  number = {3},
  pages = {216--238},
  year = {2005},
  doi = {10.26421/QIC5.3-3},
  eprint = {quant-ph/0404180},
  archivePrefix = {arXiv}
}

@article{Szalay2021,
  author = {Szalay, Szil{\'a}rd and Zimbor{\'a}s, Zolt{\'a}n and M{\'a}t{\'e}, Mih{\'a}ly and Barcza, Gergely and Schilling, Christian and Legeza, {\"O}rs},
  title = {Fermionic systems for quantum information people},
  journal = {J. Phys. A},
  volume = {54},
  number = {39},
  pages = {393001},
  year = {2021},
  doi = {10.1088/1751-8121/ac0646},
  eprint = {2006.03087},
  archivePrefix = {arXiv},
  optprimaryClass = {quant-ph}
}

@article{BoteroReznik2003,
  author = {Botero, Alonso and Reznik, Benni},
  title = {Modewise entanglement of {Gaussian} states},
  journal = {Phys. Rev. A},
  volume = {67},
  pages = {052311},
  year = {2003},
  doi = {10.1103/PhysRevA.67.052311},
  eprint = {quant-ph/0209026},
  archivePrefix = {arXiv}
}

@article{BoteroReznik2004,
  author = {Botero, Alonso and Reznik, Benni},
  title = {{BCS}-like modewise entanglement of fermion {Gaussian} states},
  journal = {Phys. Lett. A},
  volume = {331},
  number = {1--2},
  pages = {39--44},
  year = {2004},
  doi = {10.1016/j.physleta.2004.08.037},
  eprint = {quant-ph/0404176},
  archivePrefix = {arXiv}
}

@article{SpeeSchwaigerGiedkeKraus2018,
  author = {Spee, Cornelia and Schwaiger, Katharina and Giedke, G{\'e}za and Kraus, Barbara},
  title = {Mode entanglement of {Gaussian} fermionic states},
  journal = {Phys. Rev. A},
  volume = {97},
  pages = {042325},
  year = {2018},
  doi = {10.1103/PhysRevA.97.042325},
  eprint = {1712.07560},
  archivePrefix = {arXiv},
  optprimaryClass = {quant-ph}
}

@article{VerstraeteCirac2003,
  author = {Verstraete, Frank and Cirac, J. Ignacio},
  title = {Quantum nonlocality in the presence of superselection rules and data hiding protocols},
  journal = {Phys. Rev. Lett.},
  volume = {91},
  pages = {010404},
  year = {2003},
  doi = {10.1103/PhysRevLett.91.010404},
  eprint = {quant-ph/0302039},
  archivePrefix = {arXiv}
}

@article{SchuchVerstraeteCirac2004Resources,
  author = {Schuch, Norbert and Verstraete, Frank and Cirac, J. Ignacio},
  title = {Nonlocal Resources in the Presence of Superselection Rules},
  journal = {Phys. Rev. Lett.},
  volume = {92},
  pages = {087904},
  year = {2004},
  doi = {10.1103/PhysRevLett.92.087904},
  eprint = {quant-ph/0310124},
  archivePrefix = {arXiv}
}

@article{WisemanVaccaro2003,
  author = {Wiseman, H. M. and Vaccaro, J. A.},
  title = {Entanglement of indistinguishable particles shared between two parties},
  journal = {Phys. Rev. Lett.},
  volume = {91},
  pages = {097902},
  year = {2003},
  doi = {10.1103/PhysRevLett.91.097902},
  eprint = {quant-ph/0210002},
  archivePrefix = {arXiv}
}

@article{FriisLeeBruschi2013,
  author = {Friis, Nicolai and Lee, Antony R. and Bruschi, David Edward},
  title = {Fermionic-mode entanglement in quantum information},
  journal = {Phys. Rev. A},
  volume = {87},
  pages = {022338},
  year = {2013},
  doi = {10.1103/PhysRevA.87.022338},
  eprint = {1211.7217},
  archivePrefix = {arXiv},
  optprimaryClass = {quant-ph}
}

@article{PhysRevA.76.022311,
  author = {Ba{\~n}uls, Mari-Carmen and Cirac, J. Ignacio and Wolf, Michael M.},
  title = {Entanglement in fermionic systems},
  journal = {Phys. Rev. A},
  volume = {76},
  pages = {022311},
  year = {2007},
  doi = {10.1103/PhysRevA.76.022311},
  eprint = {0705.1103},
  archivePrefix = {arXiv},
  optprimaryClass = {quant-ph}
}

@article{Bennett1996MixedState,
  author = {Bennett, Charles H. and DiVincenzo, David P. and Smolin, John A. and Wootters, William K.},
  title = {Mixed-state entanglement and quantum error correction},
  journal = {Phys. Rev. A},
  volume = {54},
  pages = {3824--3851},
  year = {1996},
  doi = {10.1103/PhysRevA.54.3824}
}

@article{HaydenHorodeckiTerhal2001,
  author = {Hayden, Patrick M. and Horodecki, Micha{\l} and Terhal, Barbara M.},
  title = {The asymptotic entanglement cost of preparing a quantum state},
  journal = {J. Phys. A},
  volume = {34},
  number = {35},
  pages = {6891--6898},
  year = {2001},
  doi = {10.1088/0305-4470/34/35/314},
  eprint = {quant-ph/0008134},
  archivePrefix = {arXiv}
}

@article{Shor2004,
  author = {Shor, Peter W.},
  title = {Equivalence of additivity questions in quantum information theory},
  journal = {Comm. Math. Phys.},
  volume = {246},
  pages = {453--472},
  year = {2004},
  doi = {10.1007/s00220-003-0981-7},
  eprint = {quant-ph/0305035},
  archivePrefix = {arXiv}
}

@article{Hastings2009,
  author = {Hastings, M. B.},
  title = {Superadditivity of communication capacity using entangled inputs},
  journal = {Nature Physics},
  volume = {5},
  pages = {255--257},
  year = {2009},
  doi = {10.1038/nphys1224},
  eprint = {0809.3972},
  archivePrefix = {arXiv}
}

@article{Horodecki2009,
  author = {Horodecki, Ryszard and Horodecki, Pawe{\l} and Horodecki, Micha{\l} and Horodecki, Karol},
  title = {Quantum entanglement},
  journal = {Rev. Mod. Phys.},
  volume = {81},
  pages = {865--942},
  year = {2009},
  doi = {10.1103/RevModPhys.81.865},
  eprint = {quant-ph/0702225},
  archivePrefix = {arXiv}
}

@article{WolfGiedkeKruegerWernerCirac2004,
  author = {Wolf, M. M. and Giedke, G. and Kr{\"u}ger, O. and Werner, R. F. and Cirac, J. I.},
  title = {Gaussian entanglement of formation},
  journal = {Phys. Rev. A},
  volume = {69},
  pages = {052320},
  year = {2004},
  doi = {10.1103/PhysRevA.69.052320},
  eprint = {quant-ph/0306177},
  archivePrefix = {arXiv}
}

@article{GreplovaGiedke2018,
  author = {Greplov{\'a}, Eli{\v{s}}ka and Giedke, G{\'e}za},
  title = {Degradability of fermionic {Gaussian} channels},
  journal = {Phys. Rev. Lett.},
  volume = {121},
  pages = {200501},
  year = {2018},
  doi = {10.1103/PhysRevLett.121.200501},
  eprint = {1604.01954},
  archivePrefix = {arXiv},
  optprimaryClass = {quant-ph}
}

@article{SerafiniEisertWolf2005,
  author = {Serafini, Alessio and Eisert, Jens and Wolf, Michael M.},
  title = {Multiplicativity of maximal output purities of {Gaussian} channels under {Gaussian} inputs},
  journal = {Phys. Rev. A},
  volume = {71},
  pages = {012320},
  year = {2005},
  doi = {10.1103/PhysRevA.71.012320},
  eprint = {quant-ph/0406065},
  archivePrefix = {arXiv}
}

@article{NegariSalek2025,
  author = {Negari, Amir-Reza and Salek, Farzin},
  title = {Extendibility of Fermionic {Gaussian} States},
  journal = {Phys. Rev. Lett.},
  volume = {137},
  number = {8},
  pages = {080201},
  year = {2026},
  doi = {10.1103/65dd-448y},
  eprint = {2508.18532},
  archivePrefix = {arXiv},
  optprimaryClass = {quant-ph}
}

@article{ShapourianRyu2019,
  author = {Shapourian, Hassan and Ryu, Shinsei},
  title = {Entanglement negativity of fermions: Monotonicity, separability criterion, and classification of few-mode states},
  journal = {Phys. Rev. A},
  volume = {99},
  pages = {022310},
  year = {2019},
  doi = {10.1103/PhysRevA.99.022310},
  eprint = {1804.08637},
  archivePrefix = {arXiv},
  optprimaryClass = {quant-ph}
}

@article{Wootters1998,
  author = {Wootters, William K.},
  title = {Entanglement of formation of an arbitrary state of two qubits},
  journal = {Phys. Rev. Lett.},
  volume = {80},
  pages = {2245--2248},
  year = {1998},
  doi = {10.1103/PhysRevLett.80.2245},
  eprint = {quant-ph/9709029},
  archivePrefix = {arXiv}
}

@article{MatsumotoShimonoWinter2004,
  author = {Matsumoto, Keiji and Shimono, Toshiyuki and Winter, Andreas},
  title = {Remarks on Additivity of the {Holevo} Channel Capacity and of the Entanglement of Formation},
  journal = {Comm. Math. Phys.},
  volume = {246},
  number = {3},
  pages = {427--442},
  year = {2004},
  doi = {10.1007/s00220-003-0919-0},
  eprint = {quant-ph/0206148},
  archivePrefix = {arXiv}
}

@article{ChristandlWinter2004,
  author = {Christandl, Matthias and Winter, Andreas},
  title = {{``Squashed entanglement''}---An Additive Entanglement Measure},
  journal = {J. Math. Phys.},
  volume = {45},
  number = {3},
  pages = {829--840},
  year = {2004},
  doi = {10.1063/1.1643788},
  eprint = {quant-ph/0308088},
  archivePrefix = {arXiv}
}

@article{SchuchVerstraeteCirac2004,
  author = {Schuch, Norbert and Verstraete, Frank and Cirac, J. Ignacio},
  title = {Quantum entanglement theory in the presence of superselection rules},
  journal = {Phys. Rev. A},
  volume = {70},
  number = {4},
  pages = {042310},
  year = {2004},
  doi = {10.1103/PhysRevA.70.042310},
  eprint = {quant-ph/0404079},
  archivePrefix = {arXiv}
}

@Article{VidalNegativity,
  Author                   = {G. Vidal and R. F. Werner},
  Journal                  = {Phys. Rev. A},
  DOI={10.1103/PhysRevA.65.032314},
    title ={Computable measure of entanglement},
  Year                     = {2002},
  Pages                    = {032314},
  Volume                   = {65}
}

@phdthesis{PhD,
  Author                   = {J. Eisert},
    title ={Entanglement in quantum information theory},
  Year                     = {2001},
  note                    = {{PhD thesis, University of Potsdam}}
}

@Article{PlenioNegativity,
  Author                   = {M. B. Plenio},
  Journal                  = {Phys. Rev. Lett.},
  Year                     = {2005},
  DOI={10.1103/PhysRevLett.95.090503}, 
     title ={Logarithmic negativity: A full entanglement monotone that is not convex},
  Pages                    = {090503},
  Volume                   = {95}
}

@article{Lami2016monogamy,
   author  = {Lami, Ludovico and Hirche, Christoph and Adesso, Gerardo
              and Winter, Andreas},
   title   = {{Schur} Complement Inequalities for Covariance Matrices and
              Monogamy of Quantum Correlations},
   journal = {Phys. Rev. Lett.},
   volume  = {117},
   pages   = {220502},
   year    = {2016},
   doi     = {10.1103/PhysRevLett.117.220502},
   eprint  = {1607.05285},
   archivePrefix = {arXiv},
   optprimaryClass = {quant-ph}
 }

@article{Lami2017logdet,
   author  = {Lami, Ludovico and Hirche, Christoph and Adesso, Gerardo
              and Winter, Andreas},
   title   = {From Log-Determinant Inequalities to {G}aussian
              Entanglement via Recoverability Theory},
   journal = {IEEE Transactions on Information Theory},
   volume  = {63},
   number  = {11},
   pages   = {7553--7568},
   year    = {2017},
   doi     = {10.1109/TIT.2017.2737546},
   eprint  = {1703.06149},
   archivePrefix = {arXiv},
   optprimaryClass = {quant-ph}
 }

@article{TerhalDiVincenzo2002,
  author = {Terhal, Barbara M. and DiVincenzo, David P.},
  title = {Classical Simulation of Noninteracting-Fermion Quantum Circuits},
  journal = {Phys. Rev. A},
  volume = {65},
  pages = {032325},
  year = {2002},
  doi = {10.1103/PhysRevA.65.032325},
  eprint = {quant-ph/0108010},
  archivePrefix = {arXiv}
}

@article{GluzaKrumnowFriesdorfGogolinEisert2016,
  author = {Gluza, Marek and Krumnow, Christian and Friesdorf, Marcin and
            Gogolin, Christian and Eisert, Jens},
  title = {{Equilibration via Gaussification in fermionic lattice systems}},
  journal = {Phys. Rev. Lett.},
  volume = {117},
  pages = {190602},
  year = {2016},
  doi = {10.1103/PhysRevLett.117.190602},
  eprint = {1601.00671},
  archivePrefix = {arXiv},
  optprimaryClass = {cond-mat.quant-gas}
}

@article{BravyiGosset2017,
  author = {Bravyi, Sergey and Gosset, David},
  title = {Complexity of Quantum Impurity Problems},
  journal = {Comm. Math. Phys.},
  volume = {356},
  number = {2},
  pages = {451--500},
  year = {2017},
  doi = {10.1007/s00220-017-2976-9},
  eprint = {1609.00735},
  archivePrefix = {arXiv},
  optprimaryClass = {quant-ph}
}

@article{Dariano2014,
  author = {D'Ariano, Giacomo Mauro and Manessi, Franco and Perinotti, Paolo and Tosini, Alessandro},
  title = {Fermionic computation is non-local tomographic and violates monogamy of entanglement},
  journal = {EPL},
  volume = {107},
  number = {2},
  pages = {20009},
  year = {2014},
  doi = {10.1209/0295-5075/107/20009},
  eprint = {1307.7902},
  archivePrefix = {arXiv},
  optprimaryClass = {quant-ph}
}

@article{Dariano2014Feynman,
  author = {D'Ariano, Giacomo Mauro and Manessi, Franco and Perinotti, Paolo and Tosini, Alessandro},
  title = {The {Feynman} problem and fermionic entanglement: Fermionic theory versus qubit theory},
  journal = {Int. J. Mod. Phys. A},
  volume = {29},
  pages = {1430025},
  year = {2014},
  doi = {10.1142/S0217751X14300257},
  eprint = {1403.2674},
  archivePrefix = {arXiv},
  optprimaryClass = {quant-ph}
}

@misc{Bravyi2005Capacity,
  author = {Bravyi, Sergey},
  title = {Classical capacity of fermionic product channels},
  year = {2005},
  eprint = {quant-ph/0507282},
  archivePrefix = {arXiv}
}

@article{Debarba2020,
  author = {Debarba, Tiago and Iemini, Fernando and Giedke, G{\'e}za and Friis, Nicolai},
  title = {Teleporting quantum information encoded in fermionic modes},
  journal = {Phys. Rev. A},
  volume = {101},
  number = {5},
  pages = {052326},
  year = {2020},
  doi = {10.1103/PhysRevA.101.052326},
  eprint = {2002.08201},
  archivePrefix = {arXiv},
  optprimaryClass = {quant-ph}
}

@article{AudenaertPlenio2005,
  author = {Audenaert, Koenraad M. R. and Plenio, Martin B.},
  title = {Entanglement on mixed stabilizer states: Normal forms and reduction procedures},
  journal = {New J. Phys.},
  volume = {7},
  number = {1},
  pages = {170},
  year = {2005},
  doi = {10.1088/1367-2630/7/1/170},
  eprint = {quant-ph/0505036},
  archivePrefix = {arXiv}
}

@article{DeMeloCwiklinskiTerhal2013,
  author = {de Melo, Fernando and {\'C}wikli{\'n}ski, Piotr and Terhal, Barbara M.},
  title = {The power of noisy fermionic quantum computation},
  journal = {New J. Phys.},
  volume = {15},
  number = {1},
  pages = {013015},
  year = {2013},
  doi = {10.1088/1367-2630/15/1/013015},
  eprint = {1208.5334},
  archivePrefix = {arXiv},
  optprimaryClass = {quant-ph}
}

@article{EisertEislerZimboras2018,
  author = {Eisert, Jens and Eisler, Viktor and Zimbor{\'a}s, Zolt{\'a}n},
  title = {Entanglement negativity bounds for fermionic {Gaussian} states},
  journal = {Phys. Rev. B},
  volume = {97},
  pages = {165123},
  year = {2018},
  doi = {10.1103/PhysRevB.97.165123},
  eprint = {1611.08007},
  archivePrefix = {arXiv},
  optprimaryClass = {quant-ph}
}

@article{EislerZimboras2015,
  author = {Eisler, Viktor and Zimbor{\'a}s, Zolt{\'a}n},
  title = {On the partial transpose of fermionic {Gaussian} states},
  journal = {New J. Phys.},
  volume = {17},
  number = {5},
  pages = {053048},
  year = {2015},
  doi = {10.1088/1367-2630/17/5/053048},
  eprint = {1502.01369},
  archivePrefix = {arXiv},
  optprimaryClass = {cond-mat.stat-mech}
}

@article{BravyiGossetKoenigTemme2019,
  author = {Bravyi, Sergey and Gosset, David and Koenig, Robert and Temme, Kristan},
  title = {Approximation algorithms for quantum many-body problems},
  journal = {J. Math. Phys.},
  volume = {60},
  number = {3},
  pages = {032203},
  year = {2019},
  doi = {10.1063/1.5085428},
  eprint = {1808.01734},
  archivePrefix = {arXiv},
  optprimaryClass = {quant-ph}
}

@article{Dasenbrook2016,
  author = {Dasenbrook, David and Bowles, Joseph and Brask, Jonatan Bohr and
            Hofer, Patrick P. and Flindt, Christian and Brunner, Nicolas},
  title = {Single-electron entanglement and nonlocality},
  journal = {New J. Phys.},
  volume = {18},
  pages = {043036},
  year = {2016},
  doi = {10.1088/1367-2630/18/4/043036},
  eprint = {1511.04450},
  archivePrefix = {arXiv},
  optprimaryClass = {quant-ph}
}

@article{Clauser1969,
  author = {Clauser, John F. and Horne, Michael A. and Shimony, Abner and Holt, Richard A.},
  title = {Proposed Experiment to Test Local Hidden-Variable Theories},
  journal = {Phys. Rev. Lett.},
  volume = {23},
  pages = {880--884},
  year = {1969},
  doi = {10.1103/PhysRevLett.23.880}
}

@misc{NegariEtAl2026Approximation,
  author = {Negari, Amir-Reza and Salek, Farzin and Zimbor{\'a}s, Zolt{\'a}n and Harrow, Aram and Hayden, Patrick and Eisert, Jens},
  title = {Approximation theorems for fermionic {Gaussian} states},
  year = {2026},
  eprint = {2610.01860},
  archivePrefix = {arXiv},
  primaryClass = {quant-ph},
  url = {https://arxiv.org/abs/2610.01860v1}
}

@misc{KhanianNegari2026Purification,
  author = {Baghali Khanian, Zahra and Negari, Amir-Reza},
  title = {Entanglement of purification for {Werner} states: canonical purification and nonadditivity},
  year = {2026},
  eprint = {2610.03630},
  archivePrefix = {arXiv},
  primaryClass = {quant-ph},
  url = {https://arxiv.org/abs/2610.03630v1}
}

@misc{NegariKhanian2026RenyiPurification,
  author = {Negari, Amir-Reza and Baghali Khanian, Zahra},
  title = {{R\'enyi} Entanglement of Purification Is Non-additive},
  year = {2026},
  eprint = {2608.28897},
  archivePrefix = {arXiv},
  primaryClass = {quant-ph},
  url = {https://arxiv.org/abs/2608.28897v1}
}

@misc{FarajiKhanian2026Renyi2Purification,
  author = {Faraji, Shokoufe and Baghali Khanian, Zahra},
  title = {Additivity Results for the {R\'enyi}-2 Entanglement of Purification},
  year = {2026},
  eprint = {2605.15439},
  archivePrefix = {arXiv},
  primaryClass = {quant-ph},
  url = {https://arxiv.org/abs/2605.15439v1}
}

@misc{kalarde2026fermionsfundamentallynonlocalbosons,
      title={Fermions are fundamentally more nonlocal than Bosons}, 
      author={Fatemeh Moradi Kalarde and Sadra Boreiri and Xiangling Xu and Lucas Tendick and Salman Beigi and Paolo Perinotti and Tommaso Guaita and Marc-Olivier Renou},
      year={2026},
      eprint={2606.12363},
      archivePrefix={arXiv},
      primaryClass={quant-ph},
      url={https://arxiv.org/abs/2606.12363}, 
}
\end{document}